\documentclass[aos,preprint]{imsart}

\RequirePackage{amsthm,amsmath,amsfonts,amssymb,amsthm}
\RequirePackage[authoryear]{natbib} 
\RequirePackage[colorlinks,citecolor=blue,urlcolor=blue]{hyperref}
\RequirePackage{graphicx}
\usepackage{enumitem}

\startlocaldefs

\usepackage{amsmath,amssymb,amsfonts, mathrsfs, latexsym,natbib,upgreek,relsize}
\usepackage{bbm}
\usepackage{makecell}
\def\nn{\mathbb{N}}
\usepackage{mathrsfs}

\def\rr{\mathbb{R}}

\def\nn{\mathbb{N}}

\def\argmin{\text{argmin}}

\def\1{\mathbbm{1}}

\allowdisplaybreaks
\usepackage{capt-of}

\renewcommand{\leq}{\leqslant} 
\renewcommand{\geq}{\geqslant}

\newcommand{\eps}{\varepsilon}

\renewcommand{\iff}{\leftrightarrow}

\def\qed{ \hfill $\blacksquare$}  

\newcommand{\pdet}{\textsc{pdet}}

\newtheorem{theorem}{Theorem}[section]
\newtheorem{lemma}[theorem]{Lemma}

\theoremstyle{remark}
\newtheorem{definition}[theorem]{Definition}
\newtheorem{example}{Example}
\newtheorem{corollary}[theorem]{Corollary}
\newtheorem*{remark}{Remark}

\endlocaldefs

\begin{document}

\begin{frontmatter}
\title{Generalized Fiducial Inference on Differentiable Manifolds}

\begin{aug}
\author[A]{\fnms{Alexander C} \snm{Murph}\ead[label=e1]{acmurph@live.unc.edu}},
\author[A]{\fnms{Jan} \snm{Hannig}\ead[label=e2]{jan.hannig@unc.edu}}
\and
\author[B,C]{\fnms{Jonathan P} \snm{Williams}\ead[label=e3]{jwilli27@ncsu.edu}}
\address[A]{Department of Statistics \& Operations Research,
University of NC at Chapel Hill,
\printead{e1,e2}}

\address[B]{Department of Statistics,
NC State University, Raleigh, NC,
\printead{e3}}

\address[C]{Centre for Advanced Study, Norwegian Academy of Science and Letters}
\end{aug}

\begin{abstract}
We introduce a novel approach to inference on parameters that take values in a Riemannian manifold embedded in a Euclidean space.  Parameter spaces of this form are ubiquitous across many fields, including chemistry, physics, computer graphics, and geology.  This new approach uses generalized fiducial inference to obtain a posterior-like distribution on the manifold, without needing to know a parameterization that maps the constrained space to an unconstrained Euclidean space.  The proposed methodology, called the \textit{constrained generalized fiducial distribution} (CGFD), is obtained by using mathematical tools from Riemannian geometry.  A Bernstein-von Mises-type result for the CGFD, which provides intuition for how the desirable asymptotic qualities of the unconstrained generalized fiducial distribution are inherited by the CGFD, is provided.  To demonstrate the practical use of the CGFD, we provide three proof-of-concept examples: inference for data from a multivariate normal density with the mean parameters on a sphere, a linear logspline density estimation problem, and a reimagined approach to the AR(1) model, all of which exhibit desirable coverages via simulation.  We discuss two Markov chain Monte Carlo algorithms for the exploration of these constrained parameter spaces and adapt them for the CGFD. 
\end{abstract}
\end{frontmatter}

\section{Introduction}
Consider a manifold $\mathcal{M}$, embedded in $\rr^d$, that can be determined completely by the level set of some everywhere differentiable function $g$, i.e.,
\begin{equation}\label{eq:manifoldForm}
     \mathcal{M} = \{ \theta \in \rr^d: g(\theta) = 0 \}.
\end{equation}
Constraints of this form on parameter spaces show up across a broad range of applications.  In Molecular Dynamics, these types of manifolds are used when simulating the movement of atoms in molecules; the bond structure of these molecules determines a class of constraints, of the form in \eqref{eq:manifoldForm}, on the distance and angles between atoms \citep{fixman1974}.  The Bingham-von Mises-Fisher distribution, which can be derived by restricting a Gaussian distribution to a lower-dimensional unit hypersphere, exists on these types of manifolds.  This distribution is of particular importance for applications to animal orientation, where the movements of animals are studied according to their directions with respect to a reference location \citep{batschelet1972}.  In computer graphics, constraint functions analogous to the function $g$ show up in texture synthesis, where the structure imposed by a sample image is learned and replicated \citep{Ramanarayanan2007}.  \cite{brubaker2012} study the location of a human's joints to determine a 3D pose from 2D data; the location of these joints are constrained to exist at a certain distance determined by the known length of the person's limbs.  In geology, these sorts of constraints are found in the study of crystallographic preferred orientations, where crystal samples are reoriented according to a known standard \citep{kunze2004}.  Beyond these specific applications, classical statistical problems can also be re-expressed in terms of a constrained parameter space of the form in \eqref{eq:manifoldForm}.  A handful of these are used as proof-of-concept examples in this paper.

Statistical inference on manifolds has seen significant attention in the modern literature.  Many Frequentist approaches focus on problems where the \textit{data} are on a manifold.  \cite{bhattacharya2005} develop asymptotic theory for mean and variance estimators using data on manifolds, with special attention given to the difference between \textit{intrinsic} and \textit{extrinsic} means, i.e. since the average of points on a sphere will exist off of the sphere, should one project this onto the sphere or develop a means to average directly on the sphere?  In \cite{pelletier2005}, a kernel density estimator is developed for the case where the sample space is a non-Euclidean manifold.  In the Bayesian literature, posterior distributions are developed by first reparameterizing the problem to an unconstrained space \citep{pourzanjani2017, jauch2020}, or by taking a nonparametric approach \citep{bhattacharya2010, yang2016}.  \cite{jermyn2005} discusses the invariance issues with Bayesian maximum a posteriori and minimum mean squared error estimates for the case where the Bayesian posterior is defined on a manifold.  Many of these papers assume a priori that a manifold can be expressed by an explicitly known function, either by an isometric map from an abstract manifold to a Euclidean space, or by reparameterizing the problem to an unconstrained space.  

The aforementioned Bayesian methods often require one to perform sampling on a manifold as a means to calculate an intractable normalizing constant for the posterior distribution; fortunately, sampling on manifolds also enjoys significant attention by modern researchers \citep{girolami2011, byrne2013, ma2015, liu2016}.  \cite{diaconis2013} give an overview of modern approaches to sampling on manifolds and review the necessary theory.  The authors also recommend the sampling schemes of \cite{brubaker2012} and \cite{zappa2018}, as they have particular relevance to the method in this paper.  Further details on these two methods are provided in Section \ref{sec:CFHMC}.

The abundance of potential applications and enthusiasm in the modern literature are a clear indicators of the need for novel inferential methods on manifolds.  This paper develops a new probability distribution defined directly on a parameter space of the form in \eqref{eq:manifoldForm} by constraining and rescaling a density on the ambient space $\rr^d$.  This is distinct from most existing solutions that start from known distributions already defined to have their supports on a manifold, reparameterize the space to use distributions with supports on unconstrained spaces, or take a nonparametric approach.  The authors have found only a single other instance in the literature of concentrating an ambient density to a manifold \citep{hwang1980}.  As we will discuss in later sections, this approach is fundamentally different than the approach in this paper, and it suffers from a lack of invariance to the choice of the constraint function $g$.  

This paper introduces a fresh perspective to defining a probability distribution on \eqref{eq:manifoldForm} by way of \textit{generalized fiducial inference} (GFI) that answers issues present in the existing literature.  In general, GFI provides a means to get a posterior-like distribution on a parameter space that has many advantages over its Bayesian counterparts.  Likely the most prevalent advantage is that it does not require the choice of an often-arbitrary prior distribution.  The Bayesian answer to this issue to to choose a non-informative prior or to use an objective Bayesian approach.  The issue with non-informative priors is that, outside of Jeffrey's prior, they are not guaranteed to be invariant to a change of coordinates.  This is in contrast to GFI, which is invariant to the choice of coordinates. 

As mentioned, two major advantages of the approach in this paper are that it does not require reparameterization of the manifold and that it uses GFI to define a posterior-like distribution on the (constrained) parameter space that is coordinate-invariant and does not require one to choose a subjective prior distribution.  GFI is already a strong method for inference on manifolds; as pointed out in \cite{hannig2016}, the \textit{generalized fiducial distribution} (GFD) is invariant to the choice of parameterization.  We generalize this invariance by showing that, in the constrained case, the GFD defined according to a reparameterization of the manifold (when a reparameterization is known) is equivalent to our method, which defines the GFD on \eqref{eq:manifoldForm} only using $g$.  This development, called the \textit{constrained generalized fiducial distribution} (CGFD), inherits many of the desired qualities of the original GFD, such as asymptotic exactness of the fiducial confidence sets and desirable coverages, which are illustrated in data simulations in the sections following.  

The outline of this paper is as follows.  In Section \ref{sec:ConstrainedFiducial}, we give a brief introduction to the GFD and extend this definition on a manifold $\mathcal{M}$.  We prove that the CGFD is invariant with respect to parameterization (when a parameterization is known) and inherits desirable asymptotic qualities of the regular GFD.  In Section \ref{sec:CFHMC}, we discuss a constrained Hamiltonian that is used to simulate samples from the parameter space and modify it for the CGFD.  This is compared to a constrained random walk Metropolis algorithm.  In Section \ref{sec:examples}, we work through some examples, one of which compares a direct generalized fiducial solution for the parameterized problem to the proposed constrained solution without a parameterization.  Concluding remarks are in Section \ref{sec:conclusion}.

\section{Theory and Methods} \label{sec:ConstrainedFiducial}

\subsection{The Generalized Fiducial Distribution}

A central idea behind the generalized fiducial method is the transfer of randomness from a known random quantity $U$ onto the parameter space $\Theta$ through an established \textit{Data Generating Algorithm} (DGA), sometimes also called {\em data generating equation} or {\em association equation}.  The DGA, denoted as $\mathbf{y} = A(W,\theta)$, must relate a known random quantity and parameter to the observed data $\mathbf{y}$.  Under minor conditions on the DGA, \cite{hannig2016} perform the required transfer of randomness by defining the inverse of the DGA via the optimization problem
\begin{equation}\label{eq:FIDopt}
 Q_\textbf{y}(W)=\argmin_{\theta^{\star} \in \rr^d} {\| \textbf{y} - A(W, \theta^{\star} ) \|},
\end{equation}
where $\|\cdot\|$ can be taken to be, for instance, the $\ell_2$ norm.  
Next, for each small $\epsilon>0$, define the random variable $\theta_\epsilon^\star=Q_\textbf{y}(W_\epsilon^{\star})$,
where $W_\epsilon^{\star}$ has  distribution $F_0$ {\bf truncated} to the set 
\begin{equation}\label{eq:truncate}
 \{ W_\epsilon^{\star} : \|\mathbf{y} - A(W_\epsilon^{\star}, \theta_\epsilon^\star ) \| = \| \mathbf{y} - A(W_\epsilon^{\star}, Q_\mathbf{y}(W_\epsilon^{\star}) ) \| \leq \epsilon \}.
\end{equation}
Then, the GFD is defined as the limiting distribution on the random variable $\theta_\epsilon^\star$, given that it converges in distribution as $\epsilon\to 0$. 

When the sampling distribution of $\mathbf{y}$ is continuous, which is the only type of distribution considered in this paper, \cite{hannig2016} show the following result.
\begin{theorem}[\citealt{hannig2016}]
	\label{Jacobian}
	Under mild conditions given in the 2016 paper, the limiting distribution above has density 
	\begin{equation}\label{eq:Jacobian}
r_{\mathbf{y}} (  { \theta } ) = \frac { f( \mathbf{y} |  { \theta } ) J ( \mathbf{y} ,  { \theta } ) } { \int  f \left( \mathbf{y}  |  { \theta } ^ { \prime } \right) J \left(\mathbf{y} ,  { \theta } ^ { \prime } \right) d  { \theta } ^ { \prime } },	    
	\end{equation}
	where 
	$J ( \mathbf{y} ,  { \theta } ) = D \left( { \nabla_{  { \theta } } } A\left. ( w ,  { \theta } ) \right| _ { w = A ^ { - 1 } ( \mathbf{y} ,  { \theta } ) } \right).$
	Here $\nabla_{  { \theta }} A(w,\theta)$ is a gradient matrix computed with respect to $\theta$, and $D$ is a determinant-like operator that depends on the norm in \eqref{eq:FIDopt}, e.g., the $\ell_2$ norm leads to $D( M ) = \left( \operatorname { det } M ^ { \prime } M \right)^{ \frac { 1 } { 2 } }$. 
\end{theorem}
The primary aim of this paper is to develop and investigate a result similar to Theorem \ref{Jacobian} where the parameter of interest $\theta$ is assumed to exist on a manifold, $\mathcal{M}$.  When there exists a constraint $g$ on the parameter space, the gradient of the DGA in the Jacobian term may no longer be obtained as a direct application of elementary calculus; instead, this gradient is calculated on a \textit{curved} space, and the GFD must be updated accordingly.  The theory in the following section approaches this problem by correcting the gradient term $\nabla_\theta A$ so that all partial increments $d\theta$ are projected onto the manifold.  We use this approach to develop a formula for this constrained GFD that preserves the simple-to-use qualities of Theorem \ref{Jacobian}, while not requiring any new major assumptions beyond those on the constraint function.

\subsection{The Constrained Generalized Fiducial Distribution}

We restrict our attention to Euclidean manifolds that can be determined completely by an everywhere continuously differentiable, of all orders, constraint function $g$ on Euclidean space.  For any parameter vector $\theta \in \rr^d$, $\theta \in \mathcal{M}$ if and only if $g(\theta) = 0$, where $g \in C^\infty(\rr^d)$. The dimension of this manifold is defined implicitly using the constraint function; assuming that $t < d$, $\mathcal{M}$ is $d-t$ dimensional when $g: \rr^d \to \rr^t$ and the $t$ constraints described by $g$ are all linearly independent.  An example of a manifold of this form is the unit sphere $\{x : g(x) := ||x|| - 1 = 0 \}$.  We consider this example more closely after rigorously defining the GFD on a manifold.

First, define the following projection matrix, denoted $P_{\theta}$, that projects onto the null space of $\nabla_\theta g$ evaluated at some $\theta \in \mathcal{M}$.  Assume that the rows of $\nabla_\theta g$ are independent\footnote{This assumption amounts to the constraints being non-redundant.  This assumption can be circumvented if one is able to update the function $g$ to drop redundant constraints.}.  Then,
\begin{equation}\label{eq:schayProj} 
    P_{\theta} := I_d - (\nabla_{\theta} g)'(\nabla_{\theta} g (\nabla_{\theta} g)')^{-1} \nabla_{\theta} g,
\end{equation}
where $I_d$ is the $d$-dimensional identity matrix.  This projection matrix is the same as the one used in \cite{schay1995} to project differentials onto a constrained space. We use this projection to map the GFD Jacobian to the analogous GFD Jacobian calculated on the constrained space.  Colloquially, one can understand this approach in terms of vector subspaces.  The gradient at a point $\theta \in \mathcal{M}$ functions as a linear map from unit directional vectors $v$ to the directional derivative of $g$, i.e. the instantaneous rate of change of the function $g$ in the direction of $v$.  Since $g$ is already zero at $\theta$, we understand the constrained gradient as the gradient along a set of orthogonal directions that each leave $g$ unchanged.  This set of directional vectors must then be a basis for null$(\nabla_\theta g)$: the directions for which the instantaneous rates of change are zero.  Indeed, $P_\theta$ projects differentials onto the null space of $\nabla_\theta g$, which zeros out all directions that move off of $\mathcal{M}$.

We define the CGFD as follows:
\begin{definition}[Constrained GFD]\label{def:CJacobian}
Let $g:\rr^d \to \rr^t$ be some continuously differentiable, of all orders, constraint function that implicitly defines a $d-t$ dimensional manifold $\mathcal{M}$. Assume that $\nabla_\theta A(W,\theta)$ is full rank and that $\nabla_\theta g$ has full row rank.  The CGFD is defined as
\begin{equation}\label{eq:CJacobian}
r_{n,\mathcal{M}} (\theta|\mathbf{y}) = \frac{f(y | \theta) D^\star \left( \nabla_\theta A(w,\theta) |_{w = A^{-1}(\mathbf{y}, \theta)} P_\theta \right)}{\int_{\mathcal{M}} f(y | \theta^*) D^\star \left( \nabla_{\theta^*} A(w, \theta^*) |_{w = A^{-1}(\mathbf{y}, \theta^*)} P_{\theta^*} \right) d \lambda(\theta^*) },\quad \theta \in \mathcal{M},
\end{equation}
where $D^\star(M) = (\operatorname{pdet} M'M)^{\frac{1}{2}}$, the $\operatorname{pdet}$ operator denotes the pseudodeterminant, and the integral in the denominator is taken over the manifold $\mathcal{M}$.  Integrating over a manifold is treated more explicitly in Section \ref{sec:theories}.
\end{definition}
To calculate the above pseudodeterminant in practice, take the compact SVD such that $P_\theta = Q_\theta Q_\theta'$, and note that
\begin{equation}\label{eq:pseudojacobian}
    D^\star \left( \nabla_\theta A(w,\theta) |_{w = A^{-1}(\mathbf{y}, \theta)} P_\theta \right) = D\left( \nabla_\theta A(w,\theta) |_{w = A^{-1}(\mathbf{y}, \theta)} Q_\theta \right).
\end{equation} 
This is calculated by taking the product of the first $d-t$ eigenvalues.  This form for the CGFD has the added benefit of reducing to the equation in Theorem \ref{Jacobian} whenever there is no constraint, since the $\operatorname{pdet}$ and $\det$ operators are equivalent for square, full-rank matrices.  Further, this density is invariant to the choice of $g$. This can be seen in the following lemma:
\begin{lemma}\label{lem:constraintInvariant}
Let $h:\rr^d\to\rr^t$ be some continuously differentiable function such that $g^{-1}(\{0\}) = h^{-1}(\{0\}) = \mathcal{M}$.  Let $r_{n,g^{-1}(\{0\})}(\theta)$ and $r_{n, h^{-1}(\{0\})}(\theta)$ be the CGFDs from \eqref{eq:CJacobian}, defined using the constraint functions $g$ and $h$, respectively.  Then, for all $\theta \in \mathcal{M}$,
\[ r_{n,g^{-1}(\{0\}) }(\theta) = r_{n,h^{-1}(\{0\})}(\theta). \]
\end{lemma}
\begin{proof}
For any $\theta = (u_1,u_2) \in \mathcal{M}$, $u_1 \in \rr^{d-t}, u_2 \in \rr^t$, the implicit function theorem gives that there exists an open set $U \subset \rr^{d-t}$ around $u_1$ and a unique differentiable function $\varphi$ such that $g(u_1',\varphi(u_1')) = 0$ for all $u_1' \in U$.  Since $(u_1',\varphi(u_1'))\in\mathcal{M}$, it follows that $h(u_1',\varphi(u_1')) = 0,$ $\forall u_1' \in U$, and thus both $g$ and $h$ have the same implicit function on the set $U$.  Let $\mathfrak{G}_1$ and $\mathfrak{G}_2$ be the first $d-t$ columns and the last $t$ columns of the gradient $\nabla_\theta g$, respectively.  Since we have assumed that $\nabla g$ has full row rank, WLOG let $\mathfrak{G}_2$ be a full-rank $t\times t$ matrix.  Decompose the gradient of $g$ at the point $(u_1,u_2)$ in the following way:
\begin{align*}
\nabla_\theta g &=  \left( \mathfrak{G}_2 \right) \left( \mathfrak{G}_2^{-1} \begin{pmatrix} \mathfrak{G}_1 &; \mathfrak{G}_2  \end{pmatrix} \right) \\
&= \begin{pmatrix} \frac{\partial g_1}{\partial \theta_{d-t+1}} & \dots & \frac{\partial g_1}{\partial \theta_{d}} \\ \vdots & \ddots & \vdots \\ \frac{\partial g_t}{\partial \theta_{d-t+1}} & \dots & \frac{\partial g_t}{\partial \theta_{d}} \end{pmatrix} \begin{pmatrix}  -\frac{ \partial \varphi_1}{\partial \theta_{1}} & \dots & -\frac{ \partial \varphi_1}{\partial \theta_{d-t}} & 1 & \dots & 0  \\
\vdots & \ddots & \vdots & \vdots & \ddots & \vdots \\
 -\frac{ \partial \varphi_t}{\partial \theta_{1}} & \dots & -\frac{ \partial \varphi_t}{\partial \theta_{d-t}} & 0 & \dots & 1 
\end{pmatrix} =  \left( \mathfrak{G}_2 \right)  \begin{pmatrix}  -\nabla_{u_1} \varphi &; I_t \end{pmatrix}
\end{align*} 
since $\mathfrak{G}_2$ is an invertible matrix, where $[~;~]$ is a column concatenation.  One can derive the above matrix decomposition in terms of the implicit function by applying the chain rule to $g(u_1,\varphi(u_1)) = 0$.  Note then that the $\mathfrak{G}$ term drops out in the projection matrix calculation in \eqref{eq:schayProj} :
\begin{equation}\label{eq:ImplicitProjection}
    P_\theta = I - \begin{pmatrix} -\nabla_{u_1} \varphi &; I_t \end{pmatrix}' \left( (\nabla_{u_1} \varphi)' \nabla_{u_1} \varphi + I_t \right)^{-1} \begin{pmatrix}  -\nabla_{u_1} \varphi &; I_t \end{pmatrix}.
\end{equation}
This expression only depends on the implicit function and not on the function $g$.  The analogous argument using $h$ completes the proof.
\end{proof}

The above lemma proves that the CGFD is invariant to the form of the constraint function.  We will further show that the CGFD is, in fact, equivalent to the general GFD calculated a parameterization, whenever the explicit form of such a parameterization is known.  To motivate why these two approaches are distinct, yet (as we prove in Theorem \ref{thm:twoParameterizations}) lead to the same result, consider the following example:

\begin{example}\label{ex:equalmeans}
Consider two independent observations $X_1, X_2$ with DGA,
\begin{equation}\label{eq:equalmeans}
    X_i = A( Z_i, \theta):= \boldsymbol{\mu} + \Sigma^{1/2} Z_i,
\end{equation}
where $\boldsymbol\mu = (\mu_1, \mu_2)'$, $\Sigma^{1/2} = \mathtt{diag}(\sigma_1, \sigma_2)$, $\theta := (\boldsymbol{\mu}', \sigma_1, \sigma_2)'$, and $Z_i \overset{iid}{\sim} \mathcal{N}_2(0, I_2)$ for $i \in \{1,2\}$.  Assume that this model is constrained so that $\mu_1 = \mu_2$.  This manifold is easily expressed using the implicit constraint function $g(\theta) := \mu_2 - \mu_1$.  For a fixed $\theta \in \mathcal{M},$
\[ Q_\theta = \begin{pmatrix} 0 & 0 & \frac{1}{\sqrt{2}} \\ 0 & 0 & \frac{1}{\sqrt{2}} \\ 0 & 1 & 0 \\ 1 & 0 & 0 \end{pmatrix}, ~~~~ J_1 := \nabla_\theta A(w,\theta)\Big|_{w = A^{-1}(\mathbf{y},\theta)} = \begin{pmatrix} 1 & 0 & \frac{x_{1,1} - \mu_1}{\sigma_1} & 0 \\
0 & 1 & 0 & \frac{x_{1,2} - \mu_2}{\sigma_2} \\ 1 & 0 & \frac{x_{2,1} - \mu_1}{\sigma_1} & 0 \\
0 & 1 & 0 & \frac{x_{2,2} - \mu_2}{\sigma_2} \end{pmatrix}. \]
As an alternative to using Definition \ref{def:CJacobian}, a direct parameterization can be computed by updating the $\boldsymbol\mu:=(\mu, \mu)'$ in \eqref{eq:equalmeans} for $\mu = \mu_1 = \mu_2$ and applying \eqref{eq:Jacobian}.  This would lead to the gradient
\[J_2 := \nabla_\theta A(w,\theta)\Big|_{w = A^{-1}(\mathbf{y},\theta)} = \begin{pmatrix} 1 & \frac{x_{1,1} - \mu}{\sigma_1} & 0 \\
 1 & 0 & \frac{x_{1,2} - \mu}{\sigma_2} \\ 1  & \frac{x_{2,1} - \mu}{\sigma_1} & 0 \\
 1 & 0 & \frac{x_{2,2} - \mu}{\sigma_2} \end{pmatrix}. \]
 A direct calculation shows that the resulting scaling terms differ by a factor of $\sqrt{2}$, $\sqrt{2}D(J_1 Q_\theta)=D(J_2)$, which will cancel upon normalizing the density kernel.  
 \end{example}
 
One interesting note about the above example is that the presence of the constraint leads to $\mu$ terms in the CGFD Jacobian.  This is in contrast to the unconstrained case, where the $\mu$s in the GFD Jacobian drop out and the GFD matches the Bayesian posterior with Jeffrey's prior.  
 
 \subsection{Differentiable Level Sets as Riemannian Manifolds} \label{sec:riemannian}
Next, we use tools from differential geometry to argue that \eqref{eq:CJacobian} is the general form of the GFD on an implicitly defined manifold.
 
The level set $g^{-1}(\{ 0 \})$ of $g$ is a smooth ($C^\infty(\rr^d)$) submanifold of $\rr^d$ whenever the matrix $\nabla_\theta g$ is full rank.  We extend this nomenclature a step further by identifying $\mathcal{M}$ as an orientable \textit{Riemannian manifold} when coupled with the usual Euclidean inner product restricted to the tangent vectors of $\mathcal{M}$ \citep{lee2018}.

Besides allowing us to better understand the $P_\theta$ matrix introduced in Definition \ref{def:CJacobian}, Riemannian geometry is a valuable tool in proving the later results in this paper.  Amongst other qualities, establishing the differentiable level set as a Riemannian manifold gives that it is smooth enough such that it behaves locally like a Euclidean space.  This theory guarantees that locally there must exist a homeomorphic map between an open set around any point on the manifold and a Euclidean space (possibly of lower dimension).  Consider the pair $(\psi, \mathcal{U})$, where $\mathcal{U}\subseteq \mathcal{M}$ and $\psi:\mathcal{U}\to\rr^{d-t}$ is continuously differentiable\footnote{The homeomorphic property is extended to diffeomorphic since this manifold is in a sub-Euclidean space (see Appendix \ref{a:riemannian}).}.  This pair, called a \textit{smooth coordinate chart}, is an essential tool for the results in the following sections.  The collection of these pairs such that the $\mathcal{U}$'s cover all of $\mathcal{M}$ is called a \textit{smooth atlas}, and a smooth atlas that is not properly contained in any larger smooth atlas is called a \textit{maximal smooth atlas}.  The explanations of these terms are intentionally loose for the purpose of stating the Theorems and describing the main ideas in the following sections.  However, for the purposes of the proofs, they require a more careful and precise treatment.  For this reason, a proof that $g^{-1}(\{ 0 \})$ is a Riemannian manifold is provided in Appendix \ref{a:riemannian}, which includes the precise definitions of each of these terms.
 
\subsection{Theoretical Qualities of the Constrained GFD} \label{sec:theories}
There are two distinct perspectives on the mathematical treatment of Riemannian manifolds.  The \textit{intrinsic} point of view uses only the local structure inherent to Riemannian manifolds to perform calculations.  Alternatively, an \textit{extrinsic} perspective isometrically depends upon an isometric embedding of an abstract Riemannian manifold into a Euclidean space as a means to concretely understand it.  As proven by Nash, Riemannian manifolds can always be isometrically mapped into some Euclidean space, and therefore such an embedding must always exist \citep{nash1956}.  However, this embedding is not unique, and different embeddings can lead to different results for calculations on manifolds.  As an example, the intrinsic distance between two points on a manifold, called the \textit{Riemannian distance}, would depend upon the Riemannian metric determined when defining the manifold, while the extrinsic distance depends on the distance metric defined in the ambient Euclidean space and the specific embedding used \citep{bhattacharya2012}.

The method in this paper takes an intrinsic perspective on defining a probability distribution on a Euclidean submanifold.  The examples in the introduction that tackle the problem of defining a probability distribution on a constrained parameter space take approaches that are distinct from the approach in this paper.  Some start from existing distributions that are defined on known manifolds, such as the Bingham-von Mises distribution on the sphere.  Others reparameterize their problem to an unconstrained, lower-dimensional parameter space, then define a distribution there.  In this paper, we will write such a parameterization in terms of a smooth coordinate chart: let $(\psi, \mathcal{U})$ be a smooth coordinate chart such that $\mathcal{U}$ covers the entire manifold and the explicit form of $\psi$ is known.  Consider the following simple example of one such parameterization with polar coordinates on the sphere (excluding the poles), where $\mathcal{U} = \{ (\theta, \phi) : \theta \in (0, 2\pi), \phi \in (0, \pi) \}$ and for $u \in \mathcal{U}$, $\psi^{-1}(u) = \begin{pmatrix} \cos \theta \sin \phi, & \sin \theta \sin \phi, & \cos \phi \end{pmatrix}'$.

As discussed previously, a primary distinction of the approach in this paper as compared to the examples in the introduction is that it does not assume that a researcher knows a reparameterization on the parameter space that would allow them to essentially ignore the constraint.  When an explicit parameterization of the manifold is known, one can define the GFD directly on the constrained space by reparameterizing the DGA.  That is, using \eqref{eq:Jacobian} and the chain rule:
\begin{equation}\label{eq:directConstrained}
    r_\mathbf{y}(\theta) \propto f(\mathbf{y} | \psi^{-1} (u) ) D \left( \nabla_\theta A(w,\theta) \nabla_u \psi^{-1} (u) \right), ~~~\theta \in \mathcal{M}.
\end{equation}
It is known that the GFD is invariant to the choice of parametrization \citep{hannig2016}.  Thus, as long as one knows the form of $\psi^{-1}$ for a manifold, there is no need for any new theory.  However, in practice, such a parameterization may not be known, which is why there is a strong need for \eqref{eq:CJacobian}.  The CGFD is a more general approach to defining a GFD on a manifold, since it allows for when a constraint can only be written down implicitly.  

In the following Theorem, we show how the unnormalized kernels of the two different approaches to deriving a GFD on a manifold are equivalent when integrating locally around a fixed point.  This is done by showing that they are equivalent on open sets around any point on the manifold.  On these open sets, there is a natural choice for the direct parameterization needed in the former approach: the smooth coordinate chart that is guaranteed to exist when $\mathcal{M}$ is a Riemannian manifold.

\begin{theorem} \label{thm:twoParameterizations}
Let $\mathcal{M} := g^{-1} \{0 \}$ be the zero level set of some constraint function $g:\rr^d \to \rr^{t}$ in $C^\infty(\rr^d)$ such that $\nabla_\theta g$ has full row rank $t$.  Let $\mathcal{A} := \{ (\psi_a, \mathcal{U}_a ) \}_{a \in \mathcal{M}}$ be a maximal smooth atlas of $\mathcal{M}$.  Fix a point $\theta \in \mathcal{M}$, and choose any element of $\mathcal{A}$, $(\psi_\theta, \mathcal{U}_\theta),$ such that $\theta \in \mathcal{U}_\theta$.  Let $V\times W \subset \rr^d$ be the open set around the point $\theta$ on which the implicit function $\varphi:V \to W$ exists and is diffeomorphic\footnote{This function must exist by our assumption that $\nabla_\theta g$ has full row rank.}. Let $M_\theta$ be any open set $M_\theta \subseteq \mathcal{U}_\theta \cap (V \times W)$.  Then,
\begin{multline}\label{eq:kernelsEqual}
     \int_{\psi_\theta(M_\theta) } f(y | \psi_\theta^{-1} (u) ) D^\star\left( (\nabla_\theta A(w,\theta)) P_\theta |_{\theta = \psi_\theta^{-1}(u)}  \right)  D (\nabla_u\psi_\theta^{-1}(u)) du = \\  \int_{\psi_\theta(M_\theta)} f(y | \psi_\theta^{-1} (u) ) D^\star\left( \nabla_u A(w, \psi_\theta^{-1}(u)) \right) du,
\end{multline}
where the $D^\star(M) = (\operatorname{pdet} M'M)^{\frac{1}{2}}$ operator is the same as defined in Definition \ref{def:CJacobian} and $D(\nabla_u\psi_\theta^{-1}(u))$ should be recognized as the determinant of the Gramian matrix.
\end{theorem}
\begin{proof}
See Appendix \ref{a:twoParameterizationsProof}
\end{proof}

In Theorem \ref{thm:twoParameterizations}, the measure is taken on a lower-dimensional Euclidean space and rescaled using the Gramian matrix.  We discuss this measure in the following section and define it as a general measure on a Riemannian manifold.  For the normalization in Definition \ref{def:CJacobian}, the \textit{global} integral is needed.  In the following section, we also discuss a means to patch together local integrals like the ones in Theorem \ref{thm:twoParameterizations} in such a way that is invariant to the choice of coordinate chart.

\subsubsection{Volume measure}\label{sec:volumeMeasure}

We perform integration on manifolds using the \textit{Riemann–Lebesgue volume measure} of $\mathcal{M}$, which is a general construction that allows one to integrate in a way that is invariant to the choice of parameterization. This measure is defined for abstract manifolds.  However, it is simplified greatly in this paper's application because only sub-Euclidean manifolds are considered \citep{amann2009}.

First, we define a notion of measurability, which translates easily to a manifold by using the smooth atlas.  A set $A \subset \mathcal{M}$ is \textit{Lebesgue measurable on the manifold} if, for every $p \in A$, there exists a coordinate chart $(\psi, \mathcal{U})$ such that $p \in \mathcal{U}\cap A$ and $\psi(\mathcal{U}\cap A)$ is Lebesgue measurable.  Now, consider any coordinate chart $(\psi, \mathcal{U})$.  For any measurable set $A \subset \mathcal{U}$, define the volume measure,\[ \text{Vol}_{\mathcal{U}}(A) := \int_{\psi(A)} D(\nabla_u \psi^{-1}(u)|_{u=a}) da.  \]  As outlined in \cite{amann2009}, this volume measure is invariant to the specific choice of coordinate chart, and can be extended to an arbitrary measurable set $A$ in $\mathcal{M}$ by creating a disjoint, countable sequence $\{A_j\}_{j \in \nn}$ such that each $A_j$ is covered by a $\mathcal{U}_j$, the open set of a coordinate chart.  With this, they define the Riemann–Lebesgue volume measure of A:\[ \lambda(A) := \sum_{j=0}^\infty \text{Vol}_{\mathcal{U}_j}(A_j), \]
which is a locally finite measure.

With the Riemann-Lebesgue volume measure, we patch together the local integrals defined in Theorem \ref{thm:twoParameterizations} to a global integral over the manifold.  Let $\{ \mathcal{U}_\alpha \}_{\alpha \in \mathcal{A}}$ be a countable\footnote{Since a Riemannian manifold must be second-countable, any cover admits a countable subcover.}, indexed open cover of $\mathcal{M}$, where $\mathcal{A}$ is some index set, such that every $\mathcal{U}_\alpha$ is an open set from a coordinate chart in the smooth coordinate atlas of $\mathcal{M}$.  For each open set in this collection, Theorem \ref{thm:twoParameterizations} gives that the two approaches to defining a GFD on a manifold are equal.  

To patch together the local integrals in Theorem \ref{thm:twoParameterizations} to a global one, a partition of unity is needed.  A partition of unity is a collection of weighting functions that correct for the instances where the open sets in the above collection overlap.  With these weighting functions, we write the integral as a countable sum of local integrals over the cover $\{\mathcal{U}^*_\alpha\}_{\alpha \in \nn}$.  Using the definition from \cite{amann2009}, the collection of weighting functions $\{\zeta_i \}$ is a smooth partition of unity subordinate to $\{\mathcal{U}^*_\alpha\}_{\alpha \in \nn}$ if
\begin{enumerate}
    \item $\zeta_\alpha : \mathcal{M} \to [0,1]$ with support$(\zeta_\alpha)\subseteq \mathcal{U}_\alpha$ and $\zeta_\alpha \in C^{\infty}$ for all $\alpha \in A$;
    \item for every $p\in\mathcal{M}$, there is an open neighborhood $W$ such that support$(\zeta_\alpha)\cap V = \emptyset$ for all but finitely many $\alpha \in A$;
    \item $\sum_{\alpha \in A} \zeta_\alpha (p) = 1$ for every $p \in \mathcal{M}$.
\end{enumerate}
In the following, we apply this idea to define the marginal term used to normalize the CGFD.

Assuming that $F(\theta) := f(y | \theta ) D^\star\left( \nabla_\theta A(w,\theta) P_\theta  \right)$ is measurable on the manifold (in the same sense as the previous definition of Lebesgue measurable on the manifold), we piece together the integral in terms of a countable sum:
\begin{equation}\label{eq:manifoldIntegral}
\int_{\mathcal{M}} F d \lambda = \sum_{j=0}^\infty \int_{\mathcal{U}_j^*} \zeta_j F d \lambda(\theta) = \sum_{j=0}^\infty \int_{\psi(\mathcal{U}_j^*)} \zeta_j(\psi^{-1}(u)) F(\psi^{-1}(u)) D(\nabla_u \psi^{-1}(u)) du.
\end{equation}
This form of the integral is invariant to the choice of smooth coordinate chart and choice of smooth partition of unity.
 
Our form for the CGFD merges two general results, one in differential geometry and one in fiducial statistics, in a principled way to create a general-form fiducial probability density on a differentiable manifold.  In the following section, we explore the asymptotic qualities of the CGFD by proving a local limit theorem.

\subsection{Concentrating Measures}\label{sec:embedd}
Theorem \ref{thm:twoParameterizations} shows how the general-form method for calculating the CGFD with a projection matrix is equivalent to the GFD constructed after reparameterizing \eqref{eq:manifoldForm} to an unconstrained space.  Theorem \ref{thm:twoParameterizations} also shows that using the Lebesgue-Riemannian measure is equivalent to the method given by \cite{federer1969}, who uses the \textit{Hausdorff measure} to find the area measure of an embedded subspace.  This equivalency is outlined as a theorem in \cite{diaconis2013}.

The view of \cite{federer1969} and \cite{diaconis2013} of defining probability measures on a manifold makes use of a direct parameterization.  The method in this paper only assumes that one knows a continuously differentiable constraint function $g$ that determines the manifold, which is useful for cases where a direct parameterization is not known.  We show, in Section \ref{sec:CFHMC}, that the need for a parameterization may be further circumvented by approximating the explicit form of the CGFD using Markov chain Monte Carlo (MCMC) methods.  

An alternative way of calculating a constrained density without knowing a direct parameterization is given in \cite{hwang1980}, where the measure on a zero level set of a smooth constraint function is defined as the limit of a measure in the ambient space.  \cite{hwang1980} achieves their concentration of an ambient density onto a lower-dimensional manifold by using the kernel $\exp\left(\frac{-||g||^2}{\beta}\right)$ to concentrate a target density to a manifold as $\beta \to 0$.  For instance, we rewrite \cite{hwang1980}'s main result in the context of the focus of this paper:
\begin{theorem}[\citealt{hwang1980}]\label{thm:hwangresult}
    Let $r(\theta|\mathbf{y})$ be some GFD defined on $\rr^n$.  Suppose $\mathcal{M}:= \{ \theta: g(\theta) = 0\}$ for some $g^{\infty}(\rr^d)$.  Define the sequence of probability measures $P_\beta$ as
    \begin{equation}\label{eq:laplaceIntegral}
        \frac{dP_\beta}{r(\cdot|\mathbf{y})}(\theta) = \exp\left( \frac{-||g||^2}{\beta} \right) \left[ \int \exp\left( \frac{-||g||^2}{\beta} \right) r(\theta|\mathbf{y}) d\theta \right]^{-1}.
    \end{equation}
    Assume equations $(A1)-(A5)$ from \cite{hwang1980}, that $r(\theta|\mathbf{y})$ is not identically zero on $\mathcal{M},$ and that $\det\left( \nabla_\theta g (\nabla_\theta g)^T \right) \neq 0$.  Then, the limiting density of \eqref{eq:laplaceIntegral} as $\beta \to 0$ is
    \[ \frac{dP}{d \lambda}(u) = \frac{r(\theta | \mathbf{y})\det\left( \nabla_\theta g (\nabla_\theta g)^T \right)^{-1/2}  }{ \int_{\mathcal{M}} r(\theta | \mathbf{y})\det\left( \nabla_\theta g (\nabla_\theta g)^T \right)^{-1/2} d \lambda}.\]
\end{theorem}

Hwang's method requires that the manifold be embedded in an ambient Euclidean space to define a density constrained to a lower-dimensional manifold from a density defined on the ambient space.  This approach depends on the specific embedding used, meaning that Hwang's method represents an extrinsic perspective.  In contrast, the CGFD takes an \textit{intrinsic} perspective; although it requires that the manifold be a Euclidean submanifold, the density is calculated directly on the manifold using the local structure.  Indeed, as we have shown in Theorem \ref{thm:twoParameterizations}, the method in this paper is equivalent to defining a GFD on a local parameterized space, and then piecing these local parameterizations of manifold together. 

The extrinsic and intrinsic views of defining a probability distribution on a manifold can lead to different densities in much the same way that extrinsic and intrinsic sample means lead to different values when performing inference on manifolds \citep{bhattacharya2005}.  In the following example, we show that even in the simple case of the equal means problem, these approaches lead to different distributions.
\begin{example}\label{ex:hwang}
Consider the equal means problem from Example \ref{ex:equalmeans}. The Hwang approach for this problem gives the normalizing term, 
\[ \det\left( \nabla_\theta g (\nabla_\theta g)^T \right)^{-1/2} = 1. \]
Note that this normalizing term is not invariant to the choice of constraint function.  Define the same manifold using the constraint function $h = \mu_2^3 - \mu_1^3$.  Then,
\[ \det\left( \nabla_\theta h (\nabla_\theta h)^T \right)^{-1/2} = (9\mu_1^4 + 9\mu_2^4)^{-1/2}. \]
These are notably different from CGFD Jacobian term outlined in Example \ref{ex:equalmeans}.  This example illuminates the reason why the Hwang approach fails to be intrinsic.  According to \cite{lee2018}, a property of a manifold is intrinsic if it is preserved by isometries.  Consider the (isometric) identity map between $g^{-1}(\{0\})$ and $h^{-1}(\{0\})$; the density in Theorem \ref{thm:hwangresult} is different between these two spaces.
\end{example}
The intrinsic and extrinsic views often lead to different results, and there is ample discussion in statistics on which approach is preferable \citep{bhattacharya2005, nihat2017}.  For the densities investigated in this paper, the intrinsic perspective represented by the CGFD has the desirable advantage that it is invariant to the form of the constraint function.  The method from \cite{hwang1980} is introduced here to represent a contrasting, extrinsic perspective, and because it will become useful in proving the results in the following section.  For this reason, from this point onward, we will refer to the method from Hwang as the extrinsic density.  A further discussion of extrinsic vs. intrinsic perspectives for Riemannian manifolds can be found in \cite{nihat2017}.

\subsection{Bernstein-von Mises on a Manifold}\label{sec:limit}
In their 2014 paper, Sonderegger and Hannig develop a Bernstein-von Mises theorem for the GFD \citep{sonderegger2014}.  Not only does this result provide theoretical guarantees of asymptotic normality and asymptotic efficiency, but it also does for fiducial what the original Bernstein-von Mises theorem did for Bayesian; it shows that, asymptotically, the $1- \alpha$ approximate generalized fiducial confidence interval approximates the $1- \alpha$ frequentist confidence interval.

We determine conditions under which the local asymptotic normality consequent of a Bernstein-von Mises theorem for the GFD is carried over to the constrained case.  Note that this concept of a quality from an ambient density being inherited by a density constrained to a manifold is inherently an extrinsic idea. Thus, our extrinsic density is better suited for the explicit calculation of carrying over a Bernstien-von Mises result to a manifold, since it reduces the problem to taking an iterated limit: one concentrating the density to the manifold, and one increasing the data size.  To connect this result to the analogue for the CGFD, we show that the extrinsic density and the CGFD become indistinguishable in the limit.  The proofs of these results are technical in nature, so they are postponed to Appendix~\ref{a:proofs}. 

\begin{lemma}\label{lem:HwangVCGFD}
Let $g\in C^3(\rr^d)$ such that $\mathcal{M} = g^{-1}(\{0\})$ such that $\nabla_\theta g \neq 0$ whenever $\theta \in \mathcal{M}$.  Define the operator $q^*(s) := n^{-d/2}q(n^{-1/2}s + \theta_0)$ on a function $q$ as the change of coordinates to the local parameter, and $\theta_0 \in \mathcal{M}$.  Define 
\[ A_1 = \{s: \theta_0 + n^{-1/2}s \in (\mathcal{M}\cap C_{\theta_0}) \text{ and } ||s|| < \delta \log n \}, \]
\[ A_2 = \{s: \theta_0 + n^{-1/2}s \in (\mathcal{M}\cap C_{\theta_0}) \text{ and } ||s|| \geq \delta \log n \}, \]
\[A_3 = \{ s: \theta_0 + n^{-1/2}s \in (\mathcal{M} \backslash C_{\theta_0}) \},\]
for some compact set $C_{\theta_0}\subseteq\mathcal{M}$ that contains the true parameter and for some $0<\delta\leq1$.  We define the extrinsic density from Theorem \ref{thm:hwangresult} to be zero almost surely outside of $\mathcal{M}\cap C_{\theta_0}$.  Assume that there exist sequences $a_n, b_n$ such that the averages of the Jacobian terms uniformly converge to continuous, finite densities,
\begin{equation}\label{eq:jacAverages}
    a_n D(\nabla_\theta A(w,\theta)) \to \pi_1(\theta),~~~ b_n D^\star(\nabla_\theta A(w,\theta)P_\theta) \to \pi_2(\theta)
\end{equation} 
on $\mathcal{M} \cap C_{\theta_0}$, where $0 < \pi_1(\theta),\pi_2(\theta) < \infty$ in a neighborhood around the truth.  Assume that on $C_{\theta_0}$ the tails of the data likelihood $f_n$ and the CGFD decay at a sufficiently fast rate such that
\begin{align*}
    \frac{ \int_{ A_2 } f_n^*(s|\mathbf{y}) d\lambda(\theta_n(s)) }{  \int_{ A_1 } f_n^*(s|\mathbf{y}) d\lambda(\theta_n(s)) } \overset{P_{\theta_0}}{\to} 0,~~~~~~~~ &~~~~~~~~ b_n \frac{ \int_{ A_3 } r^*_{n,\mathcal{M}} (s|\mathbf{y})  d\lambda(\theta_n(s)) }{  \int_{ A_1 } f_n^*(s|\mathbf{y}) d\lambda(\theta_n(s)) } \overset{P_{\theta_0}}{\to} 0.
\end{align*}
Then,
\begin{equation} \label{eq:hwang_v_gfd}
 \int_{ \mathcal{M} } \left| \frac{ r_n(\theta | \mathbf{y}) \left( \det \left| \nabla_\theta g (\nabla_\theta g)^T \right| \right)^{-1/2} \1(\mathcal{M}\cap C_{\theta_0}) }{\int_{\mathcal{M} \cap C_{\theta_0} }r_n(\theta | \mathbf{y}) \left( \det \left| \nabla_\theta g (\nabla_\theta g)^T \right| \right)^{-1/2}d \lambda(\theta) } - r_{n,\mathcal{M}} (\theta|\mathbf{y})\right| d \lambda(\theta) \overset{P_{\theta_0}}{\to} 0,
 \end{equation}
 for $r_{n,\mathcal{M}}$ defined in \eqref{eq:CJacobian}.
\end{lemma}

With Lemma \ref{lem:HwangVCGFD}, we have that the Constrained GFD and the concentrated density from Hwang are equivalent in the limit.  Note that the convergence assumptions in this lemma on the Jacobian terms are satisfied for independent and identically distributed (i.i.d.) data, as outlined in Corollary \ref{cor:forIID}.  

To establish the local asymptotic normality of the Constrained GFD, we will first establish it for the extrinsic density, since this approach is better suited for arguing how local asymptotic normality is inherited from the ambient space.  In that direction, in lieu of assuming every technical assumption necessary for \cite{sonderegger2014}'s generalized fiducial Bernstein-von Mises result, we instead assume local asymptotic normality itself, and prove that it carries over to the extrinsic density with minor additional assumptions.  

\begin{lemma}\label{lem:BvMhwang}
   Let $g\in C^3(\rr^d)$ such that $\mathcal{M} = g^{-1}(\{0\})$ such that $\nabla_\theta g \neq 0$ whenever $\theta \in \mathcal{M}$.  Assume that
   \begin{equation}\label{eq:BvMassumption}
        \int_{\rr^d} \left| r_n^*(s | \mathbf{y}) - \phi^*_{0, I(\theta_0)}(s) \right| ds \overset{P_{\theta_0}}{\to} 0,
   \end{equation} 
   where the $*$ operator is a rescaling to the local parameter, $q^*(s) := n^{-d/2}q(n^{-1/2}s + \theta_0)$ for any function $q$, $\theta_0 \in \mathcal{M}$, $I(\theta_0)$ is the Fisher Information at the true parameter value, $\phi_{\theta_0, nI(\theta_0)}(\theta) = n^{d/2} \frac{\sqrt{\det | I(\theta_0)| }}{\sqrt{2\pi}} e^{-n (\theta-\theta_0 )^T I(\theta_0) (\theta -\theta_0) / 2}$.  Assume the conditions for Theorem \ref{thm:hwangresult} from \cite{hwang1980}, including that $\mathcal{M}$ is compact (see Appendix \ref{a:BvMproof}).  Assume that there exists a function $\gamma_n(s)>0$, where $\forall n,$ $\exp(-\gamma_n) \in \mathcal{L}_1$, $\int_{\rr^n} \exp(-\gamma_n(s)) ds \to 0$, and there exists some $N \in \nn$ such that for all $n\geq N$,
   \[ \left| r_n^*(s | \mathbf{y}) - \phi^*_{0, I(\theta_0)}(s) \right| \leq \exp(-\gamma_n(s)) < \infty. \]
    Assume that $\mathcal{M}$ is compact.  Let $\eta(u,v) = \theta$ be the transformation to the tubular neighborhood, and assume that the family $\gamma_{u,n}(v) := \gamma_n(\theta_n(\eta(u,v)))$ is uniformly equicontinuous at $v = 0$ over $u,n$. Then, for $H(\theta) = \left( \det \left| \nabla_\theta g (\nabla_\theta g)^T \right| \right)^{-1/2},$
   \begin{equation}\label{eq:HwangAndNormal}
   \int_{\mathcal{M} } \left| \frac{ r_n(\theta | \mathbf{y}) H(\theta) }{\int_{\mathcal{M} }r_n(\theta | y) H(\theta) d \lambda(\theta) } - \frac{ \phi_{\theta_0, nI(\theta_0)}(\theta) H(\theta)  }{\int_{\mathcal{M} }  \phi_{\theta_0, nI(\theta_0)}(\theta)H(\theta) d \lambda(\theta) }\right| d \lambda(\theta) \overset{P_{\theta_0}}{\to} 0.
\end{equation} 
\end{lemma}

With these two lemmas, we complete the aim of showing a Bernstein-von Mises result for the Constrained GFD:
\begin{theorem}\label{thm:ConstrainedVonMises}
    Assume all of the assumptions in Lemmas \ref{lem:BvMhwang} and \ref{lem:HwangVCGFD}.  Then,
\begin{align*}
   \int_{\mathcal{M} } \left| r_{n,\mathcal{M}} (\theta|\mathbf{y}) - \frac{ \phi_{\hat{\theta}_n, n I(\theta_0)}(\theta) \left( \det \left| \nabla_\theta g (\nabla_\theta g)^T \right| \right)^{-1/2}  }{\int_{\mathcal{M} }  \phi_{\hat{\theta}, nI(\theta_0)}(s)\left( \det \left| \nabla_\theta g (\nabla_\theta g)^T \right| \right)^{-1/2} d \lambda(\theta) }\right| d \lambda(\theta) \overset{P_{\theta_0}}{\to} 0.
\end{align*} 
When $\mathcal{M}$ is non-compact, this result still holds if $\nabla_\theta g$ is everywhere bounded.
\end{theorem}

The requirement that the Jacobian terms converge in Lemma \ref{lem:HwangVCGFD} is satisfied for i.i.d. data.  The is outlined in the following Corollary.

\begin{corollary} \label{cor:forIID}
Assume the data $\mathbf{y}\overset{iid}{\sim} f_{\theta_0}$, according to the likelihood model in Lemma \ref{lem:HwangVCGFD}.  Assume for $\zeta(y ) = \sup_{\theta \in \mathcal{M} } \left( ||\nabla_\theta A(w,\theta )|_{w = A^{-1}(y, \theta )} ||_{\mathcal{S}}^d \right)$ that $E_{P_{\theta_0}} \zeta(Y) < \infty$ for all $n \in \nn$, $\theta \in \mathcal{M}$, where $||\cdot||_{\mathcal{S}}$ is the spectral norm and $A(w,\theta)$ is the DGA.  Assume that $g$ is some continuously differentiable function such that $\mathcal{M}=g^{-1}(\{0\})$.  Then, the averages of the Jacobian terms in \eqref{eq:jacAverages} converge to continuous, finite densities:
\[ \frac{1}{n^{d/2}} D(\nabla_\theta A(w,\theta)) \to \pi_1(\theta),~~~ \frac{1}{n^{(d-t)/2 }} D^\star(\nabla_\theta A(w,\theta)P_\theta) \to \pi_2(\theta), \]
where $\pi_1(\theta),\pi_2(\theta)>0$ are supported on the manifold.
\end{corollary}

\section{Constrained Markov Chain Monte Carlo} \label{sec:CFHMC}

\subsection{Constrained Fiducial Hamiltonian Monte Carlo}\label{sec:brubakersAlg}
Hamiltonian Markov Chain Monte Carlo (HMC) is one of today's premier MCMC algorithms, seeing use across statistics, computer science, and chemistry, and implemented by most major numerical integration programs \citep{stan, BUGS, tensorflow, pymc3}.  At the core of HMC is the simulation of a Hamiltonian,
\begin{equation}\label{eq:classHamiltonian}
    H(\theta,p) = \mathfrak{K}(\theta,p) + \mathfrak{U}(\theta),
\end{equation}
which expresses the state of a dynamic system using its current position $\theta$ and momentum $p$.  An underlying assumption of this system is that the total energy (here, expressed as the sum of kinetic energy $\mathfrak{K}$ and potential energy $\mathfrak{U}$) is preserved as the location and momentum changes, i.e. $H(\theta(t), p(t)) = c \in \rr$, for all time $t$.  Further assumptions include time-reversibility, which requires that forward simulation can be undone by reversing the direction of time, and $p$-reversibility, which requires that reversing the direction of momentum be equivalent to reversing the direction of time \citep{reich1996}.
In a Hamiltonian system, the time evolution of the system is completely determined by Hamilton's equations:
\begin{equation} \label{eq:classicHamEquations}
\frac{d p}{d t} = -\frac{\partial H}{\partial \theta}; ~~~~~ \frac{d \theta}{d t} = \frac{\partial H}{\partial p}.
\end{equation}

HMC is a particularly useful approach for numeric integration of the CGFD because the implicit constraint function $g(\theta)$ can be encoded directly into the Hamiltonian.  Indeed, \cite{reich1996} does this by defining the following \textit{modified} Hamiltonian,
\begin{equation}\label{eq:constrainedHamiltonian}
    H_g(\theta,p) = H(\theta,p) + \frac{1}{2 \epsilon} g(\theta)' g(\theta),
\end{equation}
with $0 < \epsilon \ll 1$.  Hamilton's equations then become
\begin{align}\label{eq:constrainedHamEquations}
    \frac{d p}{d t} = -\frac{\partial H}{\partial \theta} - \frac{1}{\epsilon}(\nabla_\theta g(\theta))'g(\theta)&;& ~~~~~ \frac{d \theta}{d t} = \frac{\partial H}{\partial p}&.&
\end{align}
\cite{reich1996} modify \eqref{eq:constrainedHamEquations} even further by setting $\lambda = \frac{1}{\epsilon} g(\theta)$,
\begin{align}
    \frac{d p}{d t} = -\frac{\partial H}{\partial \theta} - (\nabla_\theta g(\theta))'\lambda&;&  \frac{d \theta}{d t} = \frac{\partial H}{\partial p}&;& \lambda = \frac{1}{\epsilon} g(\theta)
\end{align}
and sending $\epsilon \to 0$ to get the \textit{Constrained Hamiltonian System},
\begin{align}\label{eq:constrainedHamSystemFinal}
    \frac{d p}{d t} = -\frac{\partial H}{\partial \theta} - (\nabla_\theta g(\theta))'\lambda&;& ~~~~~ \frac{d \theta}{d t} = \frac{\partial H}{\partial p}&;&~~~~~0 = g(\theta).
\end{align}

In addition to the constraint described by $g(\theta) = 0$, this system has the following implicit constraint,
\begin{equation}\label{eq:hamImplicitConstraint}
    0 = (\nabla_\theta g(\theta)) \frac{\partial H}{\partial p}.
\end{equation}
One derives this implicit constraint by taking the time derivative of the first constraint, $g(\theta) = 0$, and applying the chain rule with the second equation in \eqref{eq:classHamiltonian}.

  \cite{brubaker2012} suggest a constrained HMC algorithm using the Hamiltonian
\begin{equation}\label{eq:brubakerConstrainedHamiltonian}
    H(\theta, p) = \Big(\mathfrak{K}(\theta,p)\Big) + \Big(\mathfrak{U}(\theta)\Big) :=  \left( \frac{1}{2} p' M(\theta)^{-1}p \right) + \left( \frac{1}{2} \log | M(\theta) | - \log \pi (\theta) \right) + \lambda' g(\theta),
\end{equation}
where $M(\theta)$ is a user-chosen mass matrix, $\lambda$ is a vector of Lagrange multipliers, and $\pi(\theta)$ is the unnormalized target density.  Their constrained HMC algorithm draws proposal momentum values from $\mathcal{N}\left(p_0 | 0, M(\theta_0)\right)$, based on the current state of the Hamiltonian system $(p_0, \theta_0),$ and then projects this momentum onto the constrained space such that \eqref{eq:hamImplicitConstraint} is satisfied.  Movement through the system is then simulated using the RATTLE algorithm, which involves alternating between steps of the Leapfrog algorithm and updating steps via Newton's method to ensure the constraints continue to be satisfied \citep{andersen1983}.  After simulating the Hamiltonian's movement, a Metropolis-Hastings acceptance step is taken, and the process is repeated.  For further details of Brubaker's Constrained HMC, see the original paper \cite{brubaker2012}.

The algorithm developed by Brubaker is \textit{nearly} immediately applicable to the CGFD case.  The Hamiltonian for this application would be \eqref{eq:brubakerConstrainedHamiltonian} with $\pi(\theta):= r_{n,\mathcal{M}}(\theta),$ the unnormalized CGFD from Definition \ref{def:CJacobian}.  What remains is to derive a meaningful expression for the location derivative of the log-likelihood of the CGFD, $\frac{d \log r_{n, \mathcal{M}}(\theta)}{d \theta}.$  The following theorem shows that this derivative has a simple form.
\begin{theorem}\label{thm:CGFDDerivative}
    Let $g(\theta_0)=0$ implicitly define a $d-t$ dimensional manifold $\mathcal{M}$ and assume that $g$ is in $C^\infty(\rr^{d})$ and that $\nabla_{\theta} g$ has full row rank $d-t$.  Let $r_{n,\mathcal{M}}(\theta)$ be the CGFD as defined in Definition \ref{def:CJacobian}, formed by the DGA $\mathbf{y} = A(\mathbf{w},\theta)$ and the differentiable density $f(\mathbf{y} | \theta)$.  Then,
    \[ \frac{\partial( \log r_{n,\mathcal{M} }(\theta))}{\partial \theta_i} = \frac{\partial \log f(\mathbf{y} | \theta)}{\partial \theta_i} + \frac{1}{2} \text{Tr} \left[(\nabla_\theta A P_\theta \nabla_\theta A')^+ ~\frac{\partial (\nabla_\theta A P_\theta \nabla_\theta A')}{\partial \theta_i} \right],\]
where $P_\theta$ is the projection matrix as defined in \eqref{eq:schayProj} and $(\cdot)^+$ is the Moore-Penrose pseudoinverse.
\end{theorem}
\begin{proof}
This result follows immediately from the titular result of \cite{holbrook2018}.
\end{proof}

\subsection{Constrained Metropolis-Hastings}\label{sec:zappasAlg}

While simulating a constrained HMC is an effective solution to the problem of sampling from a CGFD, it requires one to be able to compute the derivative of the CGFD, which involves a second derivative on the DGA and on the constraint function $g$.  When these calculations are overly burdensome, one can avoid them by sampling from the constrained GFD using a constrained MCMC algorithm.

Defining an MCMC algorithm on a constrained parameter space presents two challenges not present in a regular MCMC algorithm.  First, assuming one is currently at a state $x\in\mathcal{M}$, one must develop a scheme to calculate a step to a proposal value $y$ in such a way that $y \in \mathcal{M}$ and it is possible for this same scheme to step back to $x$ from $y$ (i.e. the detailed balance assumption is satisfied).  Second, one must be able to reasonably define the probability of taking a step on the manifold.  

\cite{zappa2018} suggest answers to these two questions by drawing a pre-proposal value $y'$ from a symmetric proposal distribution on the tangent plane at $x$, denoted $\mathcal{T}_x$, and projecting this value down to the manifold to the proposal value $y$.  The projection that \cite{zappa2018} define is \textit{not} the typical $l_2$ projection, which would require a second-derivative of the constraint function to satisfy detailed balance, but rather a projection along the normal space of the manifold at $x$.  This projection can be calculated using any non-linear solver that does not require a second derivative on $g$.  With this scheme, the probability of taking a step on the manifold is calculated from the probability mass of the proposal distribution centered at the point $x$, defined on $\mathcal{T}_x$.

\cite{zappa2018} point out that, to satisfy detailed balance, one must verify that the same proposal scheme is able to return from the proposed value to the initial point.  It is here that we found an error in Zappa's algorithm.  \cite{zappa2018} use a vector space argument to find the point $x'$ on $\mathcal{T}_y$ directly ``above" the point $x$, in the direction of the normal vectors at $y$.  However, they do not consider the possibility that there is another point on the manifold between $x'$ and $x$.  They propose a projection scheme that only checks whether the reverse projection succeeds in returning to the manifold, \textit{not} whether it succeeds in returning to the initial point.  It should be pointed out that this is a computational error, not an error in their mathematics.

We have implemented a corrected version of \cite{zappa2018}'s algorithm in Matlab that corrects for the error in the return projection scheme.  This version of the algorithm is used in the following sections.  The code for \cite{brubaker2012}'s algorithm is already publicly available; we made some slight modifications to it for the generalized fiducial case.  These files, like all code used in this paper, are available on this paper's \hyperlink{https://github.com/sirmurphalot/GFI_onManifolds}{Github page}.

\begin{remark}[Reversibility and Ergodicity]
Each of the two algorithms above are constructed in such a way that they are reversible, which ensures that the target distributions are invariant.  As pointed out in \cite{diaconis2013}, proving that sampling algorithms on manifolds are ergodic is non-trivial.  For \cite{brubaker2012}'s algorithm, this is done using the nice properties that are assumed when using a Hamiltonian.  For \cite{zappa2018}'s algorithm, ergodicity is a consequence of the additional assumption that their manifolds are compact, connected, and smooth.  In the examples in the following section, we apply both algorithms to both compact and non-compact problems.  
\end{remark}

\section{Examples} \label{sec:examples}
\subsection{Multivariate Normal on a Sphere}  Let $X_1, \dots, X_n \overset{iid}{\sim} \mathcal{N}_3(\boldsymbol\mu, I_3)$, where $\boldsymbol\mu = (\mu_1, \mu_2, \mu_3 )'$ must be on the unit sphere.  Written as an implicit constraint: $g(\boldsymbol\mu) = ||\boldsymbol\mu||_2 - 1 = 0.$  The DGA for this problem is the trivariate analogue of the DGA in \eqref{eq:equalmeans} with fixed $\Sigma:=I_3$.  

We derive the CGFD by a direct application of Definition \ref{def:CJacobian}.  The gradient of the DGA is an array of $n$ stacked $I_3$ identity matrices, i.e. $\nabla_{\boldsymbol\mu} A(u, \boldsymbol\mu) = \begin{pmatrix} I_3 ;& \dots &; I_3 \end{pmatrix}'$, and 
\[ P_{\boldsymbol\mu} = \begin{pmatrix} \mu_2^2 + \mu_3^2 & -\mu_1 \mu_2 & -\mu_1 \mu_3 \\
-\mu_1 \mu_2 & \mu_1^2 + \mu_3^2 & -\mu_2 \mu_3 \\
-\mu_1 \mu_3 & -\mu_2 \mu_3 & \mu_1^2 + \mu_2^2\end{pmatrix}.\]
Consequently the CGFD Jacobian term $J(\mathbf x,\mu)=n^{3/2}$.

Like the equal means problem, one can directly parameterize this constraint.  Here, this is done with polar coordinates, $\boldsymbol\mu:=(\cos \theta \sin \phi,~ \sin \theta \sin \phi,~ \cos \phi)'$, for $0 < \phi < \pi$ and $0 < \theta < 2\pi$.  

\begin{figure} \label{fig:fourSpheres} 
  \centering
  \makebox[\textwidth][c]{\includegraphics[scale=0.5]{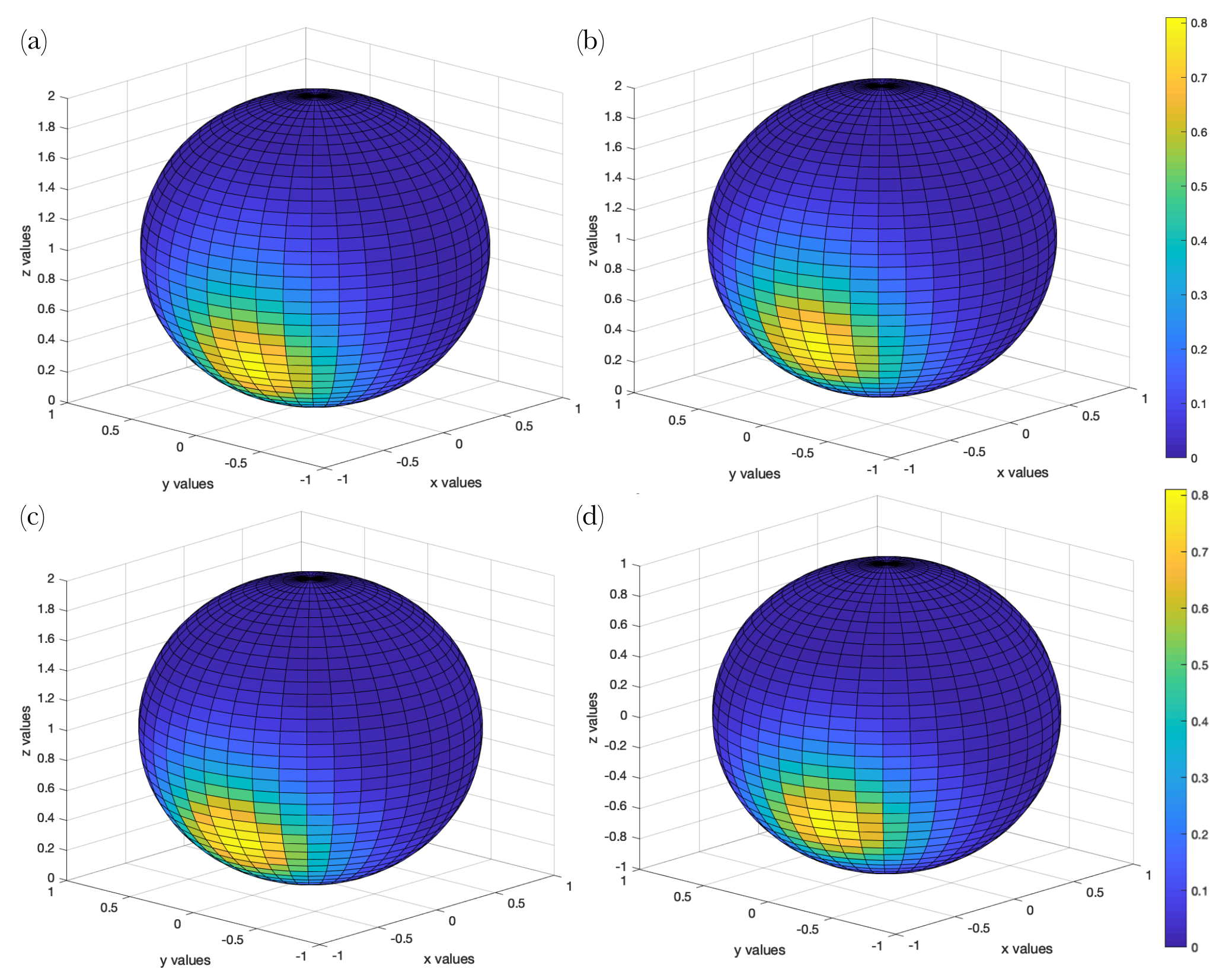}}
  \captionof{figure}{Heatmaps of the Constrained GFD distribution on a sphere, learned from a fixed data set of size $20$ where $\mu_{\text{true}} = \begin{pmatrix} -\sqrt{5/8} & \sqrt{1/8} & \sqrt{1/16} \end{pmatrix}$.  The spheres in the first row are directly calculated; the spheres in the second row are calculated using Monte Carlo methods.  $(a)$ is a graph of the probability density from the direct calculation of the Constrained GFD using Definition \ref{def:CJacobian}, $(b)$ is the usual GFD from Theorem \ref{Jacobian} for the unconstrained model using spherical coordinates, $(c)$ is the heatmap of $\sim 2e4$ HMC samples (sans half as burn-in) drawn using \cite{brubaker2012}'s algorithm with Theorem \ref{thm:CGFDDerivative}, $(d)$ is the heatmap of $2e6$ MH samples (sans half as burn-in) drawn using \cite{zappa2018}'s algorithm }
\end{figure}

Inference on the sphere is a particularly illustrating example because it shows how the direct solution of calculating a GFD with a parameterization compares to the CGFD.  This example is also used to compare the different MCMC algorithms covered in this paper.  For the purposes of this example, we generate 20 data samples from the above multivariate normal distribution with $\boldsymbol\mu = (\sqrt{5/8}, \sqrt{1/8}, \sqrt{1/16})'$.

Our first approach (sphere $(a)$ in Figure \ref{fig:fourSpheres}) leverages the fact that the constraint function is reasonably simple and $r_{n,\mathcal{M}}$ can thus be integrated numerically and with high precision.  Thus, the CGFD is calculated explicitly using this approximation of the normalizing constant.  The second approach (sphere $(b)$) reparameterizes the problem using $\boldsymbol\mu := ( \cos \theta \sin \phi, ~\sin \theta \sin \phi, ~\cos \phi)'$ for $0 \leq \theta < 2 \pi,$ $0 < \phi < \pi.$  Again, due to the relative simplicity of this example, numeric integration is available to calculate the normalizing constant of the GFD from Theorem \ref{Jacobian} explicitly.  For the last two approaches (spheres $(c)$ \& $(d)$), we simulate draws from the CGFD using the algorithms in Sections \ref{sec:brubakersAlg} \& \ref{sec:zappasAlg}, using Theorem \ref{thm:CGFDDerivative} to derive an expression for the derivative in the \cite{brubaker2012} algorithm.  For the Brubaker algorithm, we draw $\sim 2e4$ samples, and for the Zappa algorithm we draw a larger $2e6$.  The difference in number of samples was due to the fact that the Zappa algorithm was slower to converge, yet each individual MCMC step was quicker to calculate.  For each approach, half of the samples were thrown out as burn-in.  Comparing these four approaches investigates the effectiveness of the simulation algorithms while also verifying that the CGFD is invariant to different parameterizations (Theorem \ref{thm:twoParameterizations}).

Figure \ref{fig:fourSpheres} shows how all four aforementioned approaches to defining the GFD on a sphere lead to the same result.  A completely annotated step-by-step Matlab workbook of this problem, as well as all the code needed to calculate each density, is available on this paper's \hyperlink{https://github.com/sirmurphalot/GFI_onManifolds}{Github page}.  In addition to creating the visuals of Figure \ref{fig:fourSpheres}, this workbook also verifies that numerically integrating various regions of the sphere with the different approaches all lead to the same result.

\subsection{Logspline Density with Fixed Knots} \label{subsec:logsplines}

\begin{figure}
  \centering
  \makebox[\textwidth][c]{\includegraphics[scale=0.70]{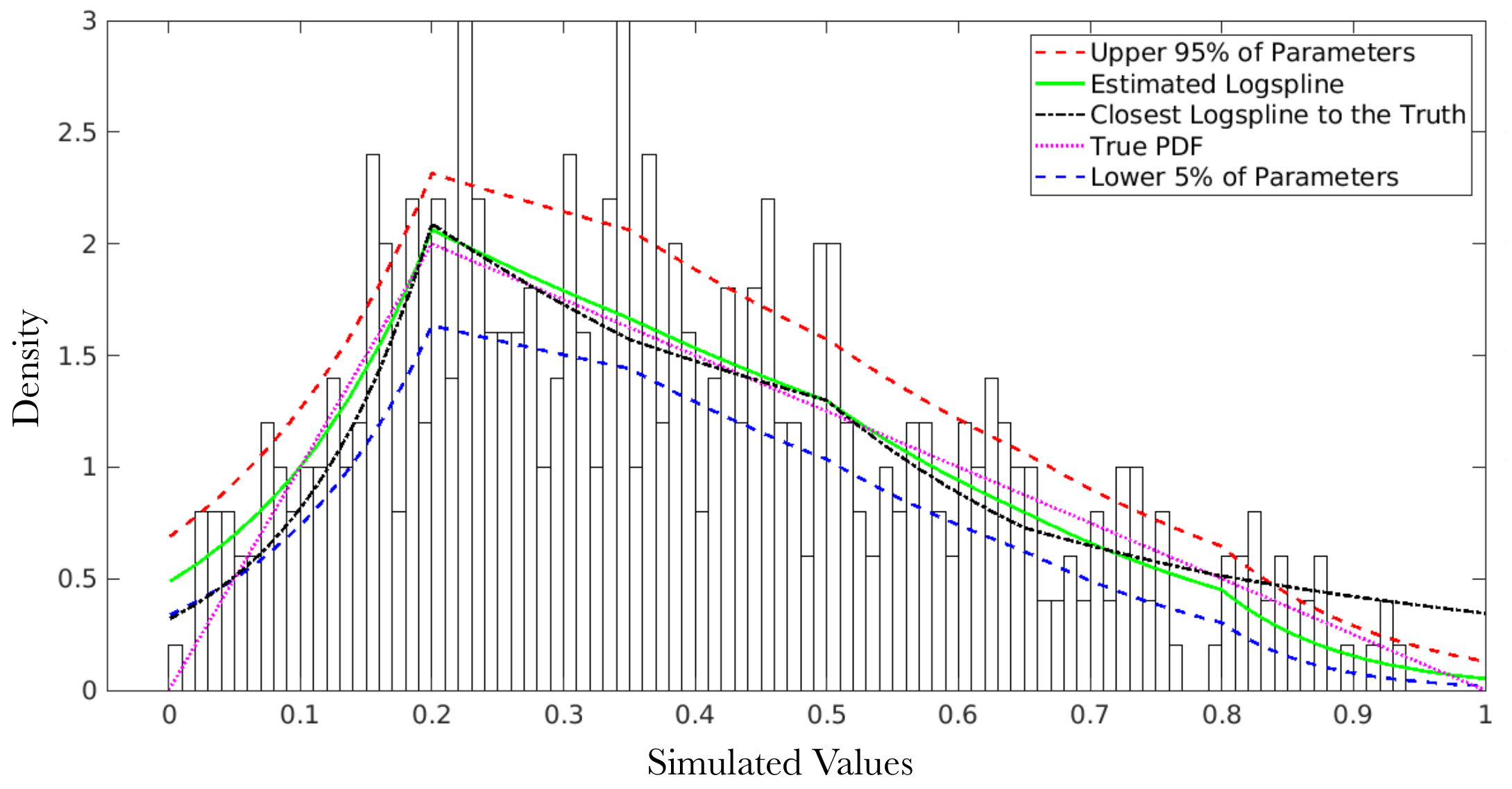}}
  \captionof{figure}{A linear logspline model with fixed knots, fit using the Constrained Metropolis-Hastings algorithm of \cite{zappa2018} and the constraint in \eqref{eq:logsplineConstraint}.  The upper and lower confidence curves need not be feasible sets of parameters.  The true PDF, and the linear logspline model that minimizes the KL divergence with the true PDF, given the set of knots, are also provided. The image is from a single iteration in the simulation study: 15000 parameters are sampled with a third thrown out as burn-in.}
  \label{fig:MHLogspline}
  \vspace{0.5cm}
  \makebox[\textwidth][c]{\includegraphics[scale=0.70]{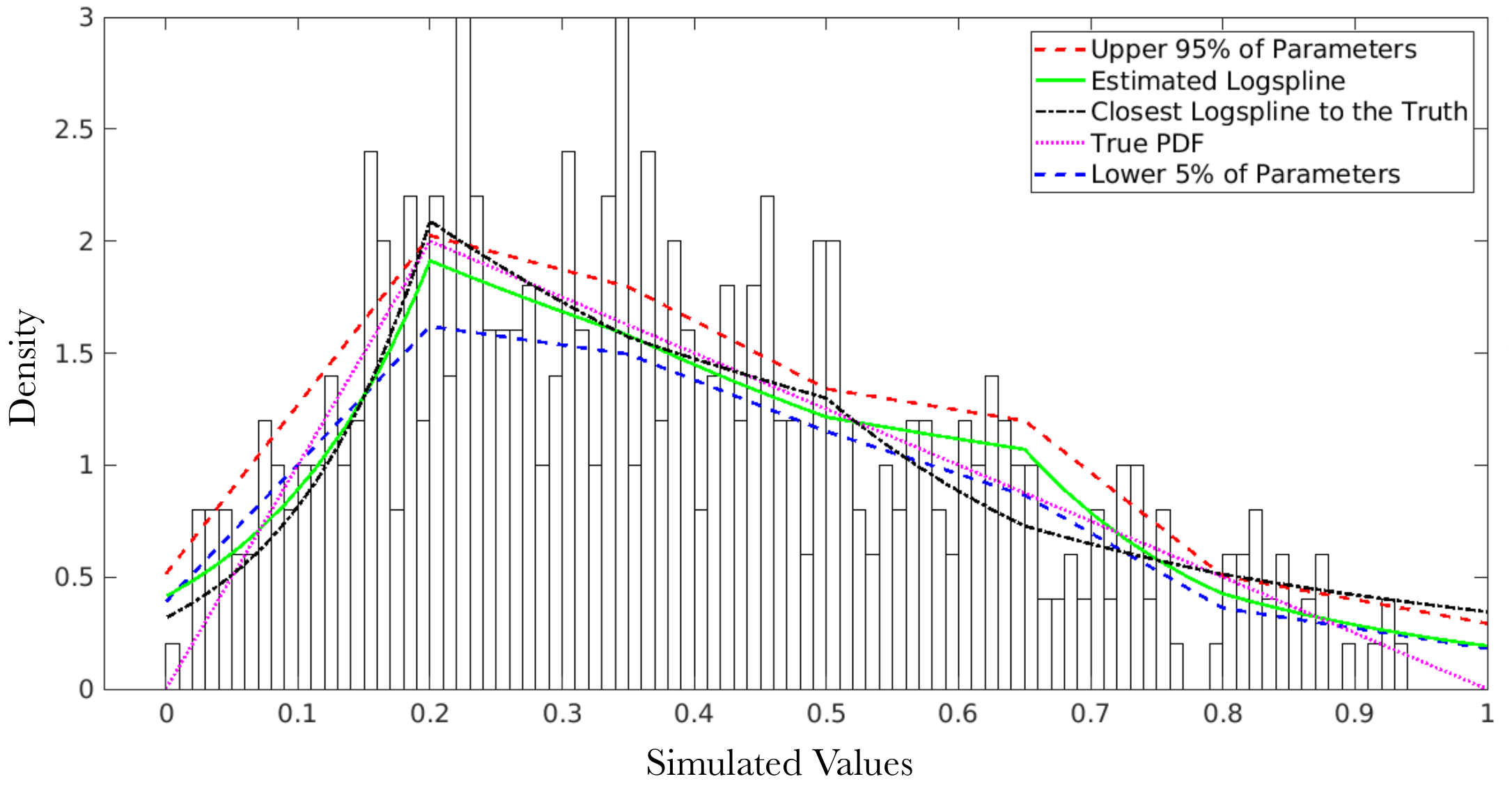}}
  \captionof{figure}{A linear logspline model with fixed knots, fit using the Constrained Hamiltonian Monte Carlo algorithm of \cite{brubaker2012} and the constraint in \eqref{eq:logsplineConstraint}.  The upper and lower confidence curves need not be feasible sets of parameters.  The true PDF, and the linear logspline model that minimizes the KL divergence with the true PDF, given the set of knots, are also provided.  This method was less flexible in exploring the parameter space when compared to the \cite{zappa2018} algorithm}
  \label{fig:HMCLogspline} 
\end{figure}

The CGFD approach provides a fresh perspective to density estimation using a logspline model with fixed knots $0 < t_1 < \dots < t_K < 1$. Consider the space of density functions on $(0,1)$ that are continuous polynomials when restricted to each of the intervals $[t_1, t_2], \dots, [t_{K-1},t_K].$  Assuming that a set of observed data $\mathbf{y}$ are generated from an unknown probability distribution with a density from this space, we perform inference on the unknown density with the $B$-splines $B_1, \dots, B_d$.  Let $\boldsymbol\theta = ( \theta_1, \dots, \theta_d )'$ be a \textit{feasible} parameter vector, which is defined to be any such parameter vector $\boldsymbol\theta \in \rr^d$ where
\[ \int_0^1 \exp \left\{ \theta_1 B_1(y) + \dots + \theta_d B_d(y) \right\}dy =1. \]
To ensure that all parameter vectors sampled are feasible, we define the logspline model as
\begin{equation}\label{eq:logsplineDensity}
    f(y|\boldsymbol\theta) = \exp \left\{ \theta_1 B_1(y) + \dots + \theta_d B_d(y) \right\}
\end{equation}
subject to the constraint
\begin{equation}\label{eq:logsplineConstraint}
   g(\boldsymbol\theta) = \log \left( \int_0^1 \exp \left\{ \theta_1 B_1(y) + \dots + \theta_d B_d(y) \right\}dy\right) = 0,
\end{equation}
which guarantees that the density integrates to 1.  For further details on the logspline model see \cite{kooperberg1991}.

We begin inference on the parameter $\boldsymbol\theta$ with the natural first step of any generalized fiducial approach: by defining the DGA. We choose
\begin{equation}\label{eq:logsplineDGA}
  y_i = F_{\boldsymbol{\theta}}^{-1}(U_i),~~~~\text{ for } i \in \{1,\dots,n\},
\end{equation}
where $F_{\boldsymbol\theta}^{-1}$ is the quantile function of \eqref{eq:logsplineDensity} and $U_i \overset{iid}{\sim} \text{Uniform}(0,1)$.  Taking the gradient of \eqref{eq:logsplineDGA} and evaluating it at $U = A^{-1}(\mathbf{y}, \boldsymbol\theta)$ leads to the gradient
\begin{equation}\label{eq:logsplineJacobian}
    \nabla_\theta A(u_i,\theta)\Big|_{u_i = A^{-1}(y_i, \theta)} = \begin{pmatrix} \frac{-\int_0^{y_i}B_j(t) \exp \left\{ \theta_1 B_1(t) + \dots + \theta_d B_d(t) \right\} ~dt }{ \exp \left\{ \theta_1 B_1(y_i) + \dots + \theta_d B_d(y_i) \right\} } \end{pmatrix}_{\{i:i \in [n]\} \times \{j:j \in [d]\} },
\end{equation}
where $[x] = \{1,2,\dots,x\}$.  Using the constraint in \eqref{eq:logsplineConstraint}, the projection matrix from Definition \ref{def:CJacobian} is calculated.  Let $c_i = \frac{E[B_i]}{\sqrt{\sum_{i=1}^d (E[B_i])^2}}$, where $E[\cdot]$ is the expectation using the density in \eqref{eq:logsplineDensity}.  Note that $c_i$ is a normalized form for the $i$th partial derivative of $g$.  Thus,
\begin{equation}\label{eq:logsplineProjection}
P_{\theta} = I_d - \mathbf{C}' \mathbf{C}; ~~~~\mathbf{C} = \begin{pmatrix} c_1 & \dots & c_d \end{pmatrix}.
\end{equation}

We now illustrate how the sampling methods discussed in this paper explore different possible \textit{feasible} parameter vectors for the $d$-dimensional logspline model. We limit the application here to linear logsplines, since for this case the integral terms in \eqref{eq:logsplineConstraint}, \eqref{eq:logsplineJacobian}, and \eqref{eq:logsplineProjection} all have a simple explicit form when using de Boor's algorithm to calculate the B-splines' polynomial coefficients.  Explicit forms also exist for quadratic logsplines in terms of large piecewise functions; the authors suspect that the cubic B-spline is also feasible, albeit beyond the scope of this example.

Assume that the observed data are generated from some unknown distribution $s \in S$.  For this example, we simulate 500 values from a triangular distribution with a peak at 0.2 and endpoints at zero and one.  Further assume that we have a fixed set of uniformly spaced knots that allow for adequate flexibility.  For this example, we chose knots that were spaced 0.15 apart.  Two additional knots were added, one before zero and one beyond one, to handle the edge cases.  This resulted in a set of 9 knots and the need to learn 7 parameter values with a single constraint.

Figure \ref{fig:MHLogspline} shows the histogram for the unknown data distribution, as well as three curves calculated from \cite{zappa2018}'s algorithm's random walk.  We simulate 150k draws from \cite{zappa2018}'s algorithm with half as burn-in.  The top curve was calculated by taking the $95^{\text{th}}$ percentile of curve values at each knot, then linearly interpolating between these points.  The bottom curve was calculated in the same way, except using the $5^{\text{th}}$ percentile.  Note that neither of these curves need be feasible.  The curve that minimizes the squared distance (on the ambient space) from all other curves is used as a point estimate.  For the set $\Theta^*$ of parameter vectors simulated by this MCMC method, the point estimate $\hat{\theta}$ is equal to 
\[ \hat{\theta} = \min_{\theta \in \Theta^*} || \theta - \bar{\theta} ||_2, \]
where $\bar{\theta}$ is the average of the values in $\Theta^*$, taken in the ambient Euclidean space.  The norm $||\cdot||_2$ is the usual $l_2$ norm.  The other two curves in Figure \ref{fig:MHLogspline} are the true PDF and the logspline curve that minimizes the KL divergence with the truth.  This last curve is an approximation of the best the linear logspline model could possibly do in estimating the true PDF, given a pre-determined set of knots and limitless data.  

Figure \ref{fig:HMCLogspline} provides the same information as Figure \ref{fig:MHLogspline}, except using the HMC algorithm instead of the Metropolis-Hastings.  We simulate 75k draws from \cite{brubaker2012}'s algorithm with half as burn-in.  Notice that, for each algorithm, the logspline model captures the variation of the observed data, which well approximates the true PDF.  

We verify the observations in Figures \ref{fig:MHLogspline} \& \ref{fig:HMCLogspline} with a full simulation study.  For each instance, observations from the same Triangular Distribution are simulated and the coverage of the sampled densities against the true PDF at the fixed knots are recorded.  In Figures \ref{fig:MH_logsplines_sim} \& \ref{fig:HMC_logsplines_sim}, we plot the nominal versus the empirical coverage of the lower confidence intervals of the true PDF value at each of the fixed knots.  These figures show how the observed coverage for these methods matches the theoretical expectations.  Performance for this method drops at the tails of the distribution, which is expected.  Furthermore, this is without loss of practical utility, since the value of the estimated PDF is typically assumed to be zero outside of the observed set of data.

\begin{figure}
\centering
\begin{minipage}{.45\textwidth}
  \centering
  \includegraphics[width=\linewidth]{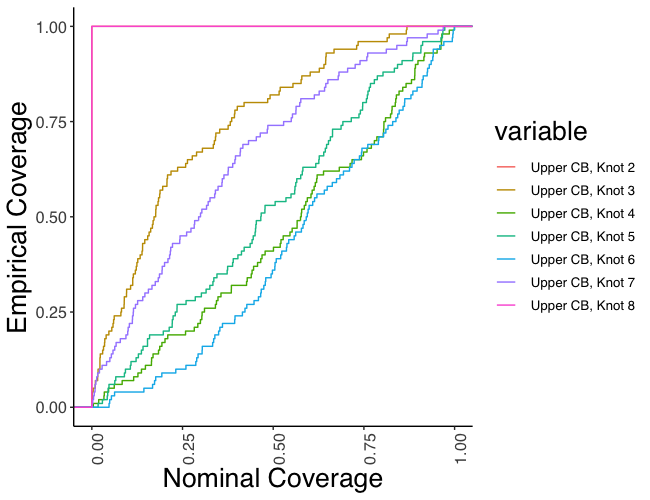}
  \captionof{figure}{Nominal versus empirical coverage for upper confidence bounds on the true PDF's value at each of the fixed knots in the logspline model, using the algorithm from \cite{zappa2018}.  This coverage is calculated over a set of 100 simulations from the triangular distribution described in Section \ref{subsec:logsplines}.  The simulated coverages matched the theoretical expectations.  While the performance drops at the knots towards the end of the distribution, this is somewhat expected due to the model set up and the sparsity of data in the tails. }
  \label{fig:MH_logsplines_sim}
\end{minipage}%
\hspace{1cm}
\begin{minipage}{.45\textwidth}
  \centering
  \includegraphics[width=\linewidth]{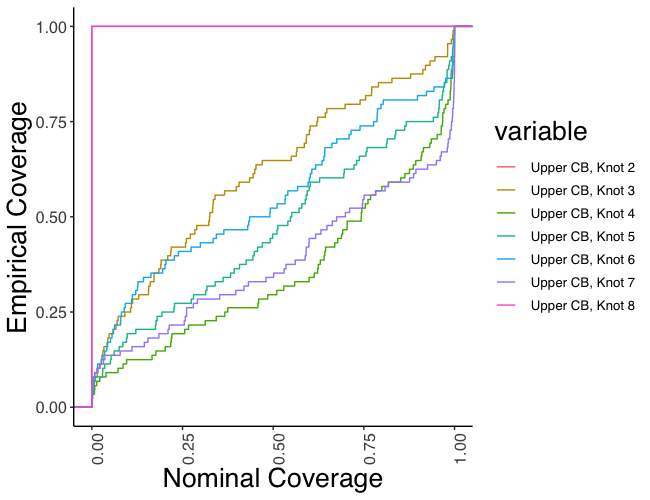}
  \captionof{figure}{Nominal versus empirical coverage for the upper confidence bounds on the true PDF's value at each of the fixed knots in the logspline model, using the algorithm from \cite{brubaker2012}.  This coverage is calculated over a set of 100 simulations from the triangular distribution described in Section \ref{subsec:logsplines}.  The simulated coverages matched the theoretical expectations.  While the performance drops at the knots towards the end of the distribution, this is somewhat expected due to the model set up and the sparsity of data in the tails.}
  \label{fig:HMC_logsplines_sim}
\end{minipage}
\end{figure}

\subsection{AR(1) Model} 

Implicit constraint functions are by no means uncommon in statistics, and many existing problems can be re-expressed in terms of such functions.  One example is with the autoregressive process of order 1 (the AR(1) model).  Consider $n$ observations from a stationary AR(1) process:
\[ X_t = \rho X_{t-1} + \epsilon,~~~~~t \in \{1,\dots,n\}, \]
where $\epsilon \sim \mathcal{N}(0, \sigma^2)$, $\sigma > 0, |\rho|\leq 1$.  The distribution of the data vector $\mathbf{X} = \begin{pmatrix}X_1 & \dots & X_n\end{pmatrix}'$ is a multivariate normal model with covariance matrix
\begin{equation} \label{eq:ar1_Sigma}
    \Sigma = \frac{\sigma^2}{1-\rho^2} \begin{pmatrix} 1 & \rho & \rho^2 & \dots & \rho^{n-1} \\ \rho & 1 & \rho & \dots & \rho^{n-2} \\ \vdots & \vdots & \vdots & \ddots & \vdots \\ \rho^{n-1} & \rho^{n-2} & \rho^{n-3} & \dots & 1  \end{pmatrix}.
\end{equation} 
With this formulation, one can re-use the DGA from \eqref{eq:equalmeans}: the DGA for a multivariate normal model with a vector of zeros as the mean.  The CGFD is well-suited for this problem as the form of this covariance matrix is determined by following constraints:
\begin{equation}\label{eq:ar1constraint}
    g(\Sigma) = \begin{pmatrix} \mathfrak{S}^1_{i,j}~  &  \mathfrak{S}^2_{k} \end{pmatrix}' = 0,~~~ \text{for $1\leq i \leq j \leq n-1$, and }\text{$1\leq k \leq n-2$,}
\end{equation}
where
\[ \mathfrak{S}^1_{i,j} = \Sigma_{i,j} - \Sigma_{i+1,j+1} \text{ and } \mathfrak{S}^2_{k} = \Sigma_{1, k+1}^2 - \Sigma_{1, k} \cdot \Sigma_{1, k+2}. \]

In the previous examples, we have had constraint functions with only a single dimension, which lead to ambient dimensions that are only one dimension larger than the dimensions of the constrained spaces.  In contrast to this, the AR(1) model is an example of a low-dimensional manifold determined by a high-dimensional constraint function,  i.e., 2 dimensional manifold and a $\frac{n(n+1)}{2}-2$ dimensional constraint function.

\cite{Murph2021} suggest reparameterizing the multivariate normal DGA from \eqref{eq:equalmeans} using the \textit{Cayley transform}, which avoids overparameterizing the problem of learning the individual elements of $\Sigma$ by instead learning the unique elements of some skew-symmetric matrix $A$ and the elements of a diagonal matrix $\Lambda$.  This new DGA is \begin{equation}\label{eq:ar1_DGA}
    \mathbf{X} = (I-A)(I+A)^{-1}\Lambda\mathbf{Z}; ~~~g^\star(A, \Lambda) := g((I-A)(I+A)^{-1}\Lambda^2 (I-A)^{-1}(I+A)).
\end{equation}
The GFD Jacobian quantity is provided by \cite{Murph2021}; they further point out that sampling of $A$ can be constrained to the hypercube $[-1,1]^{\frac{d(d-1)}{2}}$, since there must always exist a valid decomposition with the elements of the matrix $A$ on this space.  The projection matrix $P_{(A, \Lambda)}$ for this problem takes the form
\[ P_{(A, \Lambda)} = I - (\nabla g^\star)' ( \nabla g^\star (\nabla g^\star)')^{-1} \nabla g^\star~~~;~~~ \nabla g^\star =  \begin{pmatrix} \frac{\partial \mathfrak{S}^1}{\partial A} & \frac{\partial \mathfrak{S}^1}{\partial \Lambda} \\ \frac{\partial \mathfrak{S}^2}{\partial A} & \frac{\partial \mathfrak{S}^2}{\partial \Lambda}  \end{pmatrix}. \]
Each partial of each term in the vectors $\mathfrak{S}^1$ and $\mathfrak{S}^2$ is calculated using the chain rule when
\[\frac{\partial \Sigma}{\partial A_{q,r}} \text{~~~and~~~} \frac{\partial \Sigma}{\partial \Lambda_{s}},\]
are known for all $1\leq q, r \leq n$ and $1\leq s \leq n$.  These are
     $$\frac{\partial \Sigma}{\partial A_{q,r}} = 2(B + B') \text{ for } B = (I+A)^{-1}(J^{r,q} - J^{q,r} ) (I+A)^{-1} \Lambda^2 (I-A)^{-1} (I+A),$$
    $$\frac{\partial \Sigma}{\partial \Lambda_{s}} = 2\lambda_{s}(I-A)(I+A)^{-1} J^{s,s} (I-A)^{-1} (I+A),$$
where $J^{i,j}$ is a matrix of zeros with a one at position $(i,j)$.

This example illustrates some of the computational difficulties inherent to the algorithms used in this paper.   While the explicit forms for the second partial derivatives of the DGA in \eqref{eq:ar1_DGA}, needed to apply \cite{brubaker2012}'s algorithm, do exist, the computational burdens may deter most users.  A potential solution to this issue is the use of autodifferentiation software to differentiate the DGA.  This is an area for further research and development in the field of generalized fiducial inference.  Meanwhile, these experiments show that the return projection scheme in \cite{zappa2018}'s (corrected) algorithm is very sensitive to the dimension of the manifold within the ambient space.  To explore the sample space where $A$ and $\Lambda$ are constrained so that $g^\star(A,\Lambda)=0$, we perform a simulation study using \cite{zappa2018}'s algorithm only.

\begin{figure}
  \centering
  \includegraphics[width=0.5\textwidth]{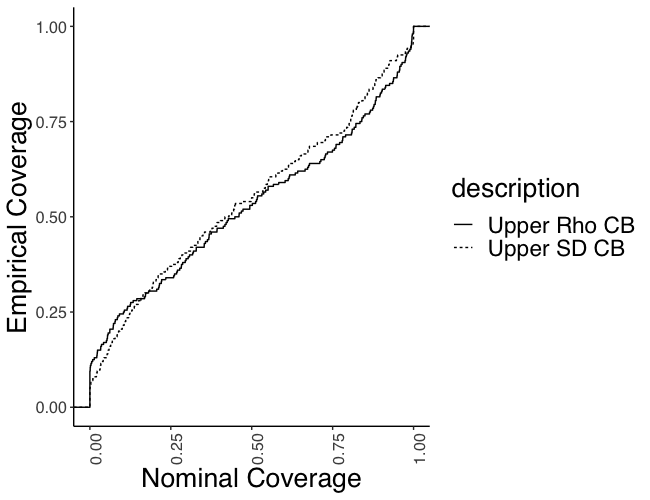}
  \captionof{figure}{Performance on the upper approximate fiducial confidence bounds on the $\rho$ and $\sigma$ parameters in the AR(1) model from \eqref{eq:ar1_DGA}, using \cite{zappa2018}'s MH algorithm. The empirical and nominal coverages mostly match.  At each instance of the simulation, we generate 10 observations for an AR(1) process, sample 500k parameters, burn-in 150k of these samples, and record the coverage of the upper approximate fiducial confidence bound. }\label{fig:ar1_sim}
\end{figure}

\begin{figure}
  \centering
  \includegraphics[width=\textwidth]{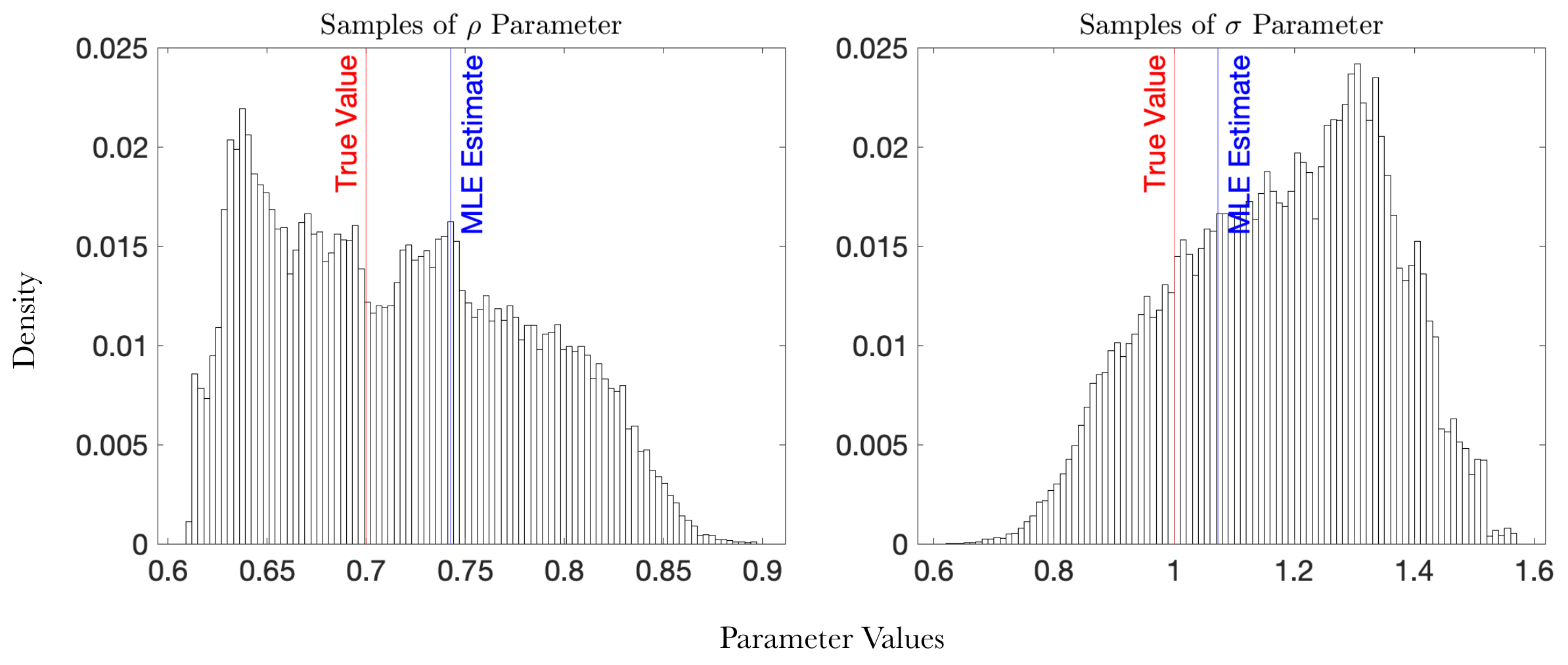}
  \captionof{figure}{Samples of the $\sigma$ parameter from \eqref{eq:ar1_Sigma} using \cite{zappa2018}'s MH algorithm, starting from the true value.  The MLE estimate is numerically approximated by maximizing the likelihood function. }\label{fig:ar1}
\end{figure}

In Figure \ref{fig:ar1}, we see that the algorithm is able to explore the area around the true value and the MLE estimates for the $\rho$ \& $\sigma$ parameters.  At each instance of the simulation, we generate 10 observations for an AR(1) process, sample 400k parameters, burn-in half of these samples, and record the coverage of the lower 95\% approximate fiducial confidence intervals.  The results for this simulation study are seen in Figure \ref{fig:ar1_sim}.  This method shows decent coverage on the lower-dimensional parameters, with slight overcoverage in $\rho$ and slight undercoverage in $\sigma$.  We investigate a single experiment in the simulation in Figure \ref{fig:ar1}.  Samples of the $\sigma$ and $\rho$ parameters show appropriate exploration of the area around the true value.  The samples for the $\rho$ parameter are somewhat skewed.  This may, in some part, be due to the computational difficulties of the reverse projection scheme.

\section{Conclusion} \label{sec:conclusion}

We have defined a broadly-applicable density on constrained parameter spaces, called the CGFD, and shown that it is a means of defining the GFD on a constrained space when one is unable to reparameterize the problem to an unconstrained space.  This form for the CGFD is invariant to how one expresses the manifold $\mathcal{M}$, which is in contrast to alternative means of defining a constrained density of a submanifold embedded in Euclidean space.  Whenever a Bernstein-von Mises-type result holds for the GFD, we have shown how this result transfers to a local limit theorem \textit{on the manifold}.  As a means to illustrate the practical use and user-friendliness, both trivial and nontrivial examples have been provided.  To aid in the explicit calculation of the GFD, two algorithms to sample values from the constrained GFD have been explored and additional tools needed for these algorithms have been developed, where necessary.  

The simulation studies performed in this paper have shown the inferential potential of the CGFD.  This being said, we believe that this method could be improved upon by using computational approaches alternative to the ones explored in this paper.  Future work would test this method on other manifold sampling algorithms, potentially using autodifferentiation software to bypass some of the computational burdens.

\begin{appendix}

\section{The Differentiable Level Set as a Riemannian Manifold}\label{a:riemannian}
This appendix serves two purposes.  The first purpose is to prove that the use of Riemannian geometry in this paper is justified.  All the geometry tools used in this paper require that the level set of a $C^\infty$ function be a Riemannian manifold.  The proof that this is true is outlined in this section.  The second purpose is to give more precise definitions to the terms \textit{Riemannian manifold}, \textit{smooth atlas}, and \textit{smooth coordinate chart}, and to enumerate specific properties of the level set used in the paper (such has the fact that it has a countable basis).

Much of the following explanation will adopt the notation in \cite{lee2013} and \cite{lee2018}.  To establish $\mathcal{M} := g^{-1}(\{0\})$ as a Riemannian manifold, we first show that it is a \textit{smooth topological manifold}.  The requirements for this structure are:
\begin{enumerate}
    \item that $\mathcal{M}$ be a Topological and Hausdorff Space;
    \item that $\mathcal{M}$ have a countable basis;
    \item that $\mathcal{M}$ be locally homeomorphic to open sets in $\rr^d$.
\end{enumerate}
$\mathcal{M}$ meets many of the requirements of this structure immediately because it is a subset of the metric space $\rr^d$.  For example, by taking the inner product metric on $\rr^d,$ the metric on $\mathcal{M}$ is simply the restriction of the inner product to the manifold.  Since $\mathcal{M}$ is then a metric space, it is immediately a topology and a Hausdorff space.  

The fact that $\mathcal{M}$ has a countable basis is also inherited from $\rr^d$.  Let $\mathcal{B}$ be the space of all balls centered at rational points that have rational radii. Then, $\mathcal{B}_\mathcal{M} = \{ B\cap \mathcal{M} : B \in \mathcal{B}\}$ satisfies the requirements for such a basis. 

For the third requirement, one requires a set of pairs $\{(U_a, \psi_a)\}_{a \in \mathcal{M}}$, for $a \in \mathcal{M}, U_a \subseteq \mathcal{M}$ and $\psi_a: U_a \to \rr^{d-t}$, called a \textit{smooth atlas}, such that 
\begin{enumerate}[label=\alph*)]
    \item the $U_a$'s cover $\mathcal{M}$;
    \item the maps $\psi_a$ are homeomorphisms;
    \item for any two pairs $(U_a,\psi_a), (U_b, \psi_b)$, either $U_a \cap U_b = \emptyset$ or $\psi_a \circ \psi_b^{-1}$ is a diffeomorphism.
\end{enumerate}
A function $f:V \to \rr^d$ for $V\subseteq{\mathcal{M}}$ is considered smooth if and only if $f\circ \psi_a:V\cap \mathcal{U}_a\to \rr^d$ for all coordinate charts in the smooth atlas. Since this smooth atlas exists is a sub-Euclidean space, it follows that each individual map $\psi_a$ is a diffeomorphism.  Each individual element in a smooth atlas is called a \textit{smooth coordinate chart}.

We prove that these conditions are met for the differentiable level set in \eqref{eq:manifoldForm} by using the implicit function theorem.  Choose any $a \in \mathcal{M}$, then there must exist an index set $\{i\}$, an open ball $B_\epsilon(a_{i_1},\dots,a_{i_{d-t}}):=U \subset \rr^{d-t}$, and the map 
\[ \varphi_a : U \to \rr^{t}, \]
which is continuously differentiable.  Define the projection map $\psi_a: \mathcal{U}_a \to \rr^{d-t}$ as $\psi_a(\theta_1,\dots,\theta_d) = \theta_{i_{1}}, \dots, \theta_{i_{d-t}}$ where $\mathcal{U}_a:=(U, \varphi_a(U)) \cap \mathcal{M}$.  The collection of homomorphism $\phi_a : \mathcal{U}_i \to \rr^{d-t}$ and open set $\mathcal{U}_a$ pairs cover $\mathcal{M}$, and are continuously differentiable.  For any two pairs $(U_a,\psi_a), (U_b, \psi_b)$ such that $U_a \cap U_b \neq \emptyset$, the map $\psi_b \circ \psi_a^{-1}: \psi_a(\mathcal{U}_a\cap\mathcal{U}_b ) \to \psi_b(\mathcal{U}_a\cap\mathcal{U}_b )$, and its inverse, are differentiable by the Chain Rule.
 
Now that $\mathcal{M}$ is established as a smooth topological manifold, what remains is to find a Riemannian metric $\mathfrak{g}_p$ that maps 
\[ \mathcal{T}_p \mathcal{M} \times \mathcal{T}_p \mathcal{M} \to \rr, \]
for every point $p \in \mathcal{M}$, where $\mathcal{T}_p$ is the collection of all tangent vectors at $p$.  Since $\mathcal{M}$ is a submanifold of the Euclidean space $\rr^d$, the standard Euclidean inner product metric restricted to the tangent spaces on $\mathcal{M}$ is used.  Thus, $\mathcal{M}$ is a smooth topological manifold with a Riemannian metric $\mathfrak{g},$ and thus $\mathcal{M}$ is a Riemannian manifold.
 
Another important property of this manifold is that it is orientable, which loosely (in Euclidean space) means that moving a coordinate basis along the manifold does not change its orientation (which could mean, for instance, the determinant changing sign).  A classic example of a non-orientable surface in Euclidean space is the M\"{o}bius strip: a normal vector moving along this surface will return to its original location, but face in the opposite direction.  This explanation of orientable is intentionally loose, since defining it precisely would require an understanding of local frames for the tangent bundle, which is a tool that is not needed for this paper.  However, orientable is a quality needed for the method of integration on $\mathcal{M}$ outlined in this paper, so it must be established.  This is immediate from Proposition 15.5 in \cite{lee2013}, and the assumption that $\nabla_\theta g \neq 0$.

For more reading on Riemannian manifolds, the authors would strongly suggest Lee's two books: \cite{lee2013} and \cite{lee2018}.

\section{Proofs of Results in Paper}\label{a:proofs}

\subsection{Proof of Theorem \ref{thm:twoParameterizations}}\label{a:twoParameterizationsProof}
By \cite{hannig2016}, we know that the GFD is invariant to the choice of smooth parameterization.  Thus, if we had two coordinate charts, $(\mathcal{U}_\theta, \psi_\theta)$ and $(\mathcal{V}_\theta, \phi_\theta)$, such that $M_\theta \subseteq \mathcal{U}_\theta\cap  \mathcal{V}_\theta$, then the right hand side of \eqref{eq:kernelsEqual} would be unchanged regardless of the choice of coordinate chart:
\begin{align*}
    \int_{\phi_\theta(M_\theta)} q_2(\phi_\theta, v)dv = \int_{\psi_\theta(M_\theta)} q_2(\psi_\theta, u)du,
\end{align*} 
where
\[ q_2(\psi_\theta, u) = f(y | \psi_\theta^{-1} (u) ) D\left( \nabla_v A(w, \psi_\theta^{-1}(u)) \right). \]
We will show that the left hand side of \eqref{eq:kernelsEqual} is also invariant to the choice of coordinate chart.  Since $(\mathcal{U}_\theta, \psi_\theta)$ and $(\mathcal{V}_\theta, \phi_\theta)$ are each elements of $\mathcal{A}$, they must be smoothly compatible.  Thus, by the chain rule,
\[ \nabla_u \psi_\theta^{-1}(u) |_u = \nabla_u \phi_\theta^{-1} \circ \phi_\theta( \psi_\theta^{-1}(u)) |_u =  \nabla_v \phi_\theta^{-1} |_{v = \phi_\theta( \psi_\theta^{-1}(u)) } \left(\nabla_u (\phi_\theta \circ \psi_\theta^{-1}(u) )|_u \right).  \]
By the fact that $\phi_\theta \circ \psi_\theta^{-1}$ is diffeomorphism, the matrix $\left(\nabla_u (\phi_\theta \circ \psi_\theta^{-1}(u) )|_u \right)$ is invertible, which leads to
\[ \nabla_u \psi_\theta^{-1}(u) |_u \left(\nabla_u (\phi_\theta \circ \psi_\theta^{-1}(u) )|_u \right)^{-1} =   \nabla_v \phi_\theta^{-1} |_{\phi_\theta \circ \psi_\theta^{-1}(u) }.  \]
For $v = \phi_\theta \circ \psi_\theta^{-1}(u)$, a change of variables to the left hand side of \eqref{eq:kernelsEqual} shows the desired equivalence:
\begin{align*}
    \int_{\phi_\theta(M_\theta) } q_1(\phi_\theta, v)  D (\nabla_v\phi_\theta^{-1}(v)) dv &= \int_{\psi_\theta(M_\theta) } q_1(\psi_\theta, u)  D (\nabla_v \phi_\theta^{-1}(v) |_{\phi_\theta \circ \psi_\theta^{-1}(u)}) \left|\nabla_u (\phi_\theta \circ \psi_\theta^{-1}(u) ) \right| du \\
    &= \int_{\psi_\theta(M_\theta) } q_1(\psi_\theta, u)  D (\nabla_u \psi_\theta^{-1}(u) |_u)  du,
\end{align*}
where
\[ q_1(\psi_\theta, u) =f(y | \psi_\theta^{-1} (u) ) D^\star\left( (\nabla_\theta A(w,\theta)) P_\theta |_{\theta = \psi_\theta^{-1}(u)}  \right). \]

We will construct a specific smooth coordinate chart (that must be in $\mathcal{A}$ since $\mathcal{A}$ is maximal) for which \eqref{eq:kernelsEqual} is satisfied.  Since each integration in \eqref{eq:kernelsEqual} is invariant to the choice of smooth coordinate chart, this is a sufficient condition for this theorem.  The construction will be done using the implicit function.  Since the gradient $\nabla_\theta g$ is assumed to have full row rank, it must have at least $t$ linearly independent columns.  WLOG, assume that the last $t$ columns of $\nabla_\theta g$ are linearly independent.  To see this notationally, fix a point $\theta = (u_1, u_2)\in \mathcal{M}$, where the split coordinates $\theta^* = (u_1^*, u_2^*)\in \mathcal{M}$ are such that $u_1^* \in \rr^{d-t}$ and $u_2^* \in \rr^{t}$.   Then $\nabla_{u_2} g$ has full column rank.  At this point, $\nabla_{u_2} g$ is non-zero, and by the implicit function theorem there exists a differentiable function $\varphi_{\theta}:V \to W$ where $V \subseteq \rr^{d-t}$ and $W \subseteq \rr^{t}$ are open.  On this domain, this function satisfies the property that $\varphi_\theta(u_1) = u_2 \iff g(u_1, u_2) =0 \iff (u_1, u_2) \in \mathcal{M}$. 

Define the diffeomorphism $\phi^{-1}_{\theta}: V \to V \times W$,
\[ \phi^{-1}_{\theta}(u_1^*) = (u_1^*, \varphi_\theta(u_1^*)')', \]
where $V\times W \subseteq \mathcal{M}$ is open.  We argue that $(V \times W, \phi_\theta)$ is a coordinate chart in $\mathcal{A}$.  Many of the requirements for this pair to be in $\mathcal{A}$ (see Appendix \ref{a:riemannian}) have already been shown by its construction; what remains is to show that it is smoothly compatible with the other elements of $\mathcal{A}$.  Let $(\mathcal{U}_\theta^\star, \psi_\theta^\star)$ be any element of $\mathcal{A}$ such that $((V \times W) \cap \mathcal{U}_\theta^\star) \neq \emptyset$.  These two coordinate charts are smoothly compatible by the chain rule: $\frac{d}{du_1} (\psi_{\theta}^\star\circ \phi^{-1}_{\theta}(u_1) )=\left( \frac{d}{d\theta^*}\psi_{\theta}^\star(\theta^*)\right) \frac{d}{du_1}\phi^{-1}_{\theta}(u_1)$, since each map is individually differentiable.  The inverse is differentiable by identical reasoning, and these coordinate charts are therefore smoothly compatible.

We will now show that \eqref{eq:kernelsEqual} is satisfied for $(V \times W, \phi_\theta)$.  Define $T_{u_1} D_{u_1} V_{u_1}$ as the compact SVD of $\nabla_{u_1} \phi_\theta^{-1}({u_1})$ and $Q_\theta$ as the compact SVD of $P_\theta := Q_\theta Q_\theta'$.  The dimensions of these decompositions are $d\times (d-t)$ for $T_{u_1}$ and $Q_\theta$ and $(d-t)\times (d-t)$ for $D_{u_1}$ and $V_{u_1}$.  We begin by rewriting the GFD Jacobian term in the right hand side of \eqref{eq:kernelsEqual} using the chain rule and the matrix decomposition of $\nabla_{u_1} \phi_\theta^{-1}$:
\begin{align*}
    D(\nabla_{u_1} &A(w,\phi_\theta^{-1}({u_1}) ) ) = D(\nabla_\theta A(w,\theta )|_{\theta = \phi_\theta^{-1}(u_1)} T_{u_1} D_{u_1} V_{u_1}) =\\
    & D(\nabla_\theta A(w,\theta )|_{\theta = \phi_\theta^{-1}(u_1)} T_{u_1}) \det(D_{u_1}) = D(\nabla_\theta A(w,\theta ) T_{u_1}) D(\nabla_{u_1} \phi_\theta^{-1}(u_1)),
\end{align*}
since $D_{u_1}$ and $V_{u_1}$ are square matrices, and $V_{u_1}$ is orthogonal.  Notice that, when using the coordinate chart $(V \times W, \phi_\theta)$ to perform integration over $M_\theta$, we can rewrite the left and right sides \eqref{eq:kernelsEqual} each in the general form
\begin{equation}
    \int_{\phi_\theta(M_\theta)} f(y | \phi_\theta^{-1}(u_1)) J(u_1) D(\nabla_{u_1} \phi_\theta^{-1}(u_1)) du_1,
\end{equation}
where 
\[J(u_1) = D^\star( (\nabla_\theta A(w,\theta)) P_\theta|_{\theta = \phi_\theta^{-1}(u_1)}) = D( (\nabla_\theta A(w,\theta)) Q_\theta |_{\theta = \phi_\theta^{-1}(u_1)} )\]
for the left hand side of \eqref{eq:kernelsEqual} and
\[J(u_1) = D(\nabla_\theta A(w,\theta) |_{\theta = \phi_\theta^{-1}(u_1)} T_{u_1} )\]
for the right hand side.  

We will prove that the two forms above for $J(u_1)$ are equal.  Recall that $P_\theta$ can be written in terms of the implicit function $\varphi_\theta$ according to \eqref{eq:ImplicitProjection} in the proof of Lemma \ref{lem:constraintInvariant}.  By this equivalence, $P_\theta$ is a projection onto the nullspace of $\begin{bmatrix}  -\nabla_{u_1} \varphi &; I_t \end{bmatrix}$.  Notice that
\[ \begin{bmatrix}  -\nabla_{u_1} \varphi &; I_t \end{bmatrix} \nabla_{u_1} \phi^{-1}_{\theta} = \begin{bmatrix}  -\nabla_{u_1} \varphi &; I_t \end{bmatrix} \begin{bmatrix}  I_t \\ \nabla_{u_1} \varphi \end{bmatrix} = 0, \]
and thus every column of $\nabla_{u_1} \phi^{-1}_{\theta}$ is in $\text{col}(P_\theta)$.
These two spaces have the same rank, so we conclude that
\[ \text{col}(\nabla_{u_1} \phi^{-1}_{\theta}) = \text{col}(P_\theta). \]

Thus, the columns of $Q_\theta$ span the same space as the columns of $T_{u_1}$, and therefore, for some orthogonal change of basis matrix $\mathcal{E}^*$, we have $Q_\theta = T_{u_1} \mathcal{E}^*$.  Rewrite the CGFD Jacobian term as follows:
\begin{align*}
\sqrt{\det \left(Q_\theta' (\nabla_\theta A(w,\theta ))' ( (\nabla_\theta A(w,\theta)) Q_\theta \right)} &= \sqrt{ \det \left((\mathcal{E}^*)' T_{u_1}' (\nabla_\theta A(w,\theta) )' (\nabla_\theta A(w,\theta) T_{u_1}\mathcal{E}^* \right)} \\
 &=  \sqrt{\det \left(T_{u_1}' (\nabla_\theta A(w,\theta) )' (\nabla_\theta A(w,\theta) T_{u_1}\right)}.
 \end{align*}
Thus, $\eqref{eq:kernelsEqual}$ is satisfied for $(V \times W, \phi_\theta)$.  This completes the result.

\subsection{Proof of Lemma \ref{lem:HwangVCGFD}}\label{a:Lem1Proof}
Consider the projection matrix $P_\theta$ from \eqref{eq:schayProj} as a function of $\theta$.  By the assumption that the function $g$ is continuously differentiable, and that $\nabla_\theta g \neq 0$ on the manifold, the function $P_{\theta}$ is continuous.  Recall that $\theta_n(s) = \theta_0 +n^{-1/2}s$ is the reverse transformation back from the local parameter.  Then for fixed $s$, $P_{\theta_n(s)}\to P_{\theta_0}$ and $H(\theta_n(s)) \to H(\theta_0)$, where $H(\theta) = \left| \det\left(  \nabla_\theta g (\nabla_\theta g)^T  \right)\right|^{-1/2}$ and $n$ is the data size.

Recall the sets $A_1, A_2$, and $A_3$ from the Lemma statement.  On $A_1$, the Jacobian terms and the extrinsic density's normalizing term go to constants:
\[ \left| \det\left(  \nabla_\theta g (\nabla_\theta g)^T |_{\theta = \theta_n(s) }  \right)\right|^{-1/2} \to \left| \det\left(  \nabla_{\theta_0} g (\nabla_{\theta_0} g)^T |_{\theta = \theta_0 } \right)\right|^{-1/2} \]
\[a_n D\left(\nabla_\theta A(w, \theta_n(s))\right)   \to \pi_1(\theta_0)  \]
\[ b_n D^\star\left(\nabla_\theta A(w, \theta_n(s)) P_{\theta_0} \right)   \to \pi_2(\theta_0)  \]

Let $\mathcal{H}_n(s) = D(\nabla A(w,\theta_n(s))) H(\theta_n(s))$ be the multiplicative terms in the extrinsic density and let $\mathcal{C}_n(s) = D^\star(\nabla A(w,\theta_n(s))P_{\theta_n(s)})$ be the Jacobian term from the Constrained GFD.  On $A_1$ and $A_2$, the terms $\mathcal{C}_n(s)$ and $\mathcal{H}_n(s)$ are eventually bounded away from zero due to the convergence assumptions on the Jacobian terms and because $\nabla g$ is full-rank on the manifold.

We calculate on $A_1$,
\begin{multline*}
\int_{ A_1 } \left| \frac{ r^*(s | \mathbf{y}) H(\theta_n(s)) }{\int_{\mathcal{M} }r^*( \mathbbmss{s} | \mathbf{y}) H(\theta_n(\mathbbmss{s})) d \lambda(\theta_n(\mathbbmss{s})) } - \frac{f^*(s|\mathbf{y})\mathcal{C}_n(s)  }{\int_{\mathcal{M}} f^*(\mathbbmss{s}|\mathbf{y})\mathcal{C}_n(\mathbbmss{s}) d \lambda(\theta_n(\mathbbmss{s}))  } \right| d \lambda(\theta_n(s))\\
\hspace{.3cm}\leq M_n \sup_{s \in A_1 } \left| \frac{  \mathcal{H}_n(s) }{ \int_{\mathcal{M} }r^*(\mathbbmss{s} | \mathbf{y}) H(\theta_n(\mathbbmss{s})) d \lambda(\theta_n(\mathbbmss{s})) }  - \frac{ \mathcal{C}_n(s) }{\int_{\mathcal{M}} f^*(\mathbbmss{s}|\mathbf{y})\mathcal{C}_n(\mathbbmss{s}) d \lambda(\theta_n(\mathbbmss{s}))} \right|,
\end{multline*}
where $M_n = \int_{A_1} f^*(s|\mathbf{y}) d\lambda(\theta_n(s))$.  The term on the right hand side of the above inequality is bounded above by
\[\max\left\{ \sup_{s\in A_1}B_1(s) - \inf_{s\in A_1} B_4(s),\sup_{s\in A_1} B_3(s) - \inf_{s\in A_1} B_2(s)\right\}\overset{P_{\theta_0}}{\to} 1-1 = 0,\]
where
\begin{align*}
    B_1(s)&=\frac{ a_n \mathcal{H}_n(s) }{(V_n / M_n) \inf_{\mathbbmss{s} \in A_2 } a_n \mathcal{H}_n(\mathbbmss{s}) + \inf_{\mathbbmss{s} \in A_1 } a_n \mathcal{H}_n(\mathbbmss{s}) },\\
    B_2(s)&=\frac{ a_n \mathcal{H}_n(s) }{(V_n / M_n)  \sup_{\mathbbmss{s} \in A_2 } a_n \mathcal{H}_n(\mathbbmss{s}) + \sup_{\mathbbmss{s} \in A_1 } a_n \mathcal{H}_n(\mathbbmss{s}) },\\
    B_3(s) &= \frac{b_n \mathcal{C}_n(s) }{ (W_n/M_n)b_n+ (V_n/M_n)\inf_{\mathbbmss{s} \in A_2 } b_n \mathcal{C}_n(\mathbbmss{s}) + \inf_{\mathbbmss{s} \in A_1 } b_n \mathcal{C}_n(\mathbbmss{s})  },\\
    B_4(s) &= \frac{b_n \mathcal{C}_n(s) }{ (W_n/M_n) b_n+ (V_n/M_n)\sup_{\mathbbmss{s} \in A_2 } b_n \mathcal{C}_n(\mathbbmss{s}) + \sup_{\mathbbmss{s} \in A_1 } b_n \mathcal{C}_n(\mathbbmss{s})  }.
\end{align*}
 Recall that $V_n = \int_{A_2} f^*(s|\mathbf{y}) d\lambda(\theta_n(s))$ and $W_n = \int_{A_3} f^*(s|\mathbf{y})D^\star(\nabla A(w,\theta_n(s))P_{\theta_n(s)}) d \lambda(\theta_n(s))$.  Thus, we have proven that the integral of $A_1$ goes to zero.

Next, we consider $A_2$.  Following the same steps as before, the integral with respect to $A_2$,
\begin{multline*}
\int_{ A_2 } \left| \frac{ r^*(s | \mathbf{y}) H(\theta_n(s)) }{\int_{\mathcal{M} }r^*( \mathbbmss{s} | \mathbf{y}) H(\theta_n(\mathbbmss{s})) d \lambda(\theta_n(\mathbbmss{s})) } - \frac{f^*(s|\mathbf{y})\mathcal{C}_n(s)  }{\int_{\mathcal{M}} f^*(\mathbbmss{s}|\mathbf{y})\mathcal{C}_n(\mathbbmss{s}) d \lambda(\theta_n(\mathbbmss{s}))  } \right| d \lambda(\theta_n(s))
\end{multline*}
is bounded above by
\[\frac{V_n}{M_n} \max\left\{ \sup_{s\in A_2}B_1(s) - \inf_{s\in A_2} B_4(s),\sup_{s\in A_2} B_3(s) - \inf_{s\in A_2} B_2(s)\right\}\overset{P_{\theta_0}}{\to} 1-1 = 0.\]
which goes to zero in $P_{\theta_0}$ probability since $\inf_{s \in A_1 } b_n \mathcal{C}_n(s)$ and $\inf_{s \in A_1 } a_n \mathcal{H}_n(s)$ are bounded away from zero.  

Finally, since the extrinsic density is zero outside of the compact space, bound the same integral over $A_3$ by
\begin{align*}
    & \frac{ (W_n/M_n)b_n  }{ (W_n/M_n)b_n + (V_n/M_n)\inf_{\mathbbmss{s} \in A_2 } b_n \mathcal{C}_n(\mathbbmss{s}) + \inf_{\mathbbmss{s} \in A_1 } b_n \mathcal{C}_n(\mathbbmss{s}) }, 
\end{align*}
which goes to zero by assumption as well.

\subsection{Proof of Lemma \ref{lem:BvMhwang}}\label{a:BvMproof}
For completeness, we review the conditions for Theorem 3.1 from \cite{hwang1980}.  Using the notation in this lemma, Hwang assumes that that (for all $n\in\nn$):
\begin{enumerate}
    \item $\int_{ \frac{1}{2}||g(\theta)||^2<a}r_n(\theta | \mathbf{y}) d\theta >0$ if $a > \inf_\theta \frac{1}{2} ||g(\theta)||^2$;
    \item $\min_\theta \frac{1}{2} ||g(\theta)||^2$ exists and equals 0;
    \item $\int_\mathcal{M} r_n(\theta | \mathbf{y}) d\theta = 0$, the ambient measure of the manifold is zero;
    \item there exists $\epsilon > 0$ such that $\{ \frac{1}{2}||g(\theta)||^2 \leq \epsilon \}$ is compact;
    \item $g(\theta) \in C^3(\rr^d)$ and $\nabla_\theta g \neq 0$ on $\mathcal{M}$;
    \item $r_n(\theta|\mathbf{y})$ is continuous and not identically zero on $\mathcal{M}$ for all $n$.
\end{enumerate}
When not otherwise specified, the norm $||\cdot||$ is taken to mean the vector 2-norm.

Start by considering the integral
\begin{equation*}
 \mathcal{I}_n(\beta) = \frac{\int_{\mathcal{S}_{n}} \exp \left( \frac{-n||g(\theta_n(s) )||^2}{\beta} \right) \left| r_n^*(s | \mathbf{y}) - \phi^*_{0, I(\theta_0)}(s) \right| ds }{ \int_{\mathcal{S}_{n}}\exp\left( \frac{-n||g(\theta_n(s'))||^2}{\beta} \right) \phi^*_{0, I(\theta_0)}(s') ds' } ,
\end{equation*}
where $\mathcal{S}_{n} = \{ s: n^{-1/2}s + \hat{\theta}_n \in T(\epsilon)\}$, and $T(\epsilon)$ is the $\epsilon$-tubular neighborhood around $\mathcal{M}$ (see, for instance, \cite{weyl1939} for an explanation of the tubular neighborhood), chosen to be closed.  For the purposes of this proof, we require that $\mathcal{I}_n(\beta)$ go to zero in the iterated limit, and that the limits with respect to $n$ and with respect to $\beta$ be interchangeable.  This fact is established in supporting Lemma \ref{lem:supportingLemma}.

We use an argument based on the Tubular Neighborhood Theorem (TBT).  Fix any point $\theta \in \mathcal{M}$ and let $B_\theta(\epsilon)$ be the $\epsilon$-ball around $\theta$ in the ambient space.  By the TBT (see \cite{weyl1939}), any point $p \in B_\theta(\epsilon)$ can be written as a linear combination of local coordinates and normal vectors on the manifold.  Assume WLOG that $\epsilon$ is small enough such that there exists a smooth coordinate chart $(\psi_\theta, \mathcal{U}_\theta)$ such that $\mathcal{M} \cap B_\theta(\epsilon) \subseteq \mathcal{U}_\theta$.  Then,
\[ p = \eta(u, v) = \psi^{-1}_\theta(u) + vQ_{\theta}^\perp, \]
for $v \in \rr^{t}$. The TBT gives that $\eta$ is a diffeomorphism for sufficiently small $\epsilon_1$.  Note that since $\eta$ is a diffeomorphism on a compact set, it is therefore also uniformly continuous.  For ease of notation, assume that the parameterization is global, $\psi^{-1}_\theta = \psi^{-1}$.  The general case is handled with a partition of unity, which is guaranteed to exist.  Define the operator $r_n^\triangle(u,v):= r(\eta(u, v)),$ and let $J(u,v)$ be the change-of-variables Jacobian term for this transformation according to \cite{weyl1939}, which depends on the second fundamental form of the manifold and where $J(u,v) = 1$ when $v= 0$.

Rewrite $\mathcal{I}_n(\beta)$ using the TBT,
\begin{align*}\label{eq:tubularintegral}
  \frac{  \int_{\psi(\mathcal{M}) } \int_{\{|v| \leq \epsilon_1 \}} \exp \left( \frac{-n||g^\triangle(u,v)||^2}{\beta} \right) \left| r_n^\triangle(u,v | \mathbf{y}) - \phi^\triangle_{\theta_n, I(\theta_0)}(u,v) \right|J(u,v) dv d \lambda( \psi^{-1}(u) ) }{ \int_{\psi(\mathcal{M}) } \int_{\{|v'| \leq \epsilon_1 \}}\exp\left( \frac{-n||g^\triangle(\theta',v')||^2}{\beta} \right) \phi^\triangle_{\theta_n, I(\theta_0)}(\theta',v')J(u',v')  dv' d \lambda( \psi^{-1}(u'))},
\end{align*}

In the following, define $\mathbb{M}(q):= \int_\mathcal{M} q \cdot \det \left(\left| \nabla_\theta g (\nabla_\theta g)^T \right| \right)^{-1/2} d \lambda(\theta)$ for a density $q$.  Since the limits with respect to $\beta$ and $n$ are interchangeable, we may send $\beta\to 0$ without issue.  Thus, we rewrite $\mathcal{I}_n(0):=\lim_{\beta\to0} \mathcal{I}_n(\beta)$ as
\begin{align*}
   \int_\mathcal{M} \biggl|  \zeta_n(0) \frac{r_n(\theta | \mathbf{y}) \left( \det \left| \nabla_\theta g (\nabla_\theta g)^T \right| \right)^{-1/2}}{\mathbb{M}(r_n(\theta | \mathbf{y})) }  - \frac{\phi_{\theta_0, nI(\theta_0)}(\theta) \left( \det \left| \nabla_\theta g (\nabla_\theta g)^T \right| \right)^{-1/2} }{\mathbb{M}(\phi_{\theta_0, nI(\theta_0)})} \biggr| d \lambda(\theta)
\end{align*}
where
\[ \zeta_n(\beta) = \frac{\int_{\psi(\mathcal{M}) } \int_{\{|v'| \leq \epsilon_1 \}}\exp\left( \frac{-n||g^\triangle(\theta',v')||^2}{\beta} \right) r_n^\triangle(\theta',v') J(\psi(\theta'),v') d v' d \lambda( \theta') }{\int_{\psi(\mathcal{M}) } \int_{\{|v'| \leq \epsilon_1 \}} \exp \left( \frac{-n||g^\triangle(\theta',v')||^2}{\beta} \right) \phi^\triangle_{\theta_n, I(\theta_0)}(\theta',v') J(\psi(\theta'),v') dv' d \lambda( \theta') }.\]
Notice that $\zeta_n(0):= \lim_{\beta\to 0} \zeta_n(\beta) = \frac{\mathbb{M}(r_n)}{\mathbb{M}(\phi_{\theta_0, nI(\theta_0)})}$.  Using the triangle inequality, bound $\zeta_n(\beta)$ using $\mathcal{I}_n(\beta)$,
\[  1 - \mathcal{I}_n(\beta) \leq \zeta_n(\beta) \leq 1 + \mathcal{I}_n(\beta).\]
As proven in Lemma \ref{lem:supportingLemma}, $\mathcal{I}_n(\beta)$ goes to zero in the iterated limit, thus we swap the $\beta$ and $n$ limits of $\zeta_n(\beta)$ without issue, and get that $\lim_{n\to\infty}\lim_{\beta \to 0} \zeta_n(\beta ) = \lim_{\beta\to 0} \lim_{n\to\infty} \zeta_n(\beta)= 1$.  

Take the double limit of $\mathcal{I}_n(\beta)$ to get that $\lim_{n\to\infty} \lim_{\beta \to 0} \mathcal{I}_n(\beta) =\lim_{\beta \to 0} \lim_{n\to\infty}  \mathcal{I}_n(\beta) = \lim_{n\to\infty} \mathcal{I}_n(0)$, and note that \eqref{eq:HwangAndNormal} from the theorem statement is estimated by
\begin{align*}
    &\hspace{-0.5cm}\int_{\mathcal{M} } \left| \frac{r_n(\theta | \mathbf{y}) \left( \det \left| \nabla_\theta g (\nabla_\theta g)^T \right| \right)^{-1/2}}{\mathbb{M}(r_n)}  - \frac{\phi_{\theta_0, nI(\theta_0)}(\theta) \left( \det \left| \nabla_\theta g (\nabla_\theta g)^T \right| \right)^{-1/2}}{\mathbb{M}(\phi_{\theta_0, nI(\theta_0)})}   \right| d \lambda(\theta)\\
    &\hspace{0cm} = \int_{\mathcal{M} } \biggl| \frac{ r_n(\theta | \mathbf{y}) \left( \det \left| \nabla_\theta g (\nabla_\theta g)^T \right| \right)^{-1/2} }{{\mathbb{M}(r_n)}} -  \zeta_n(0) \frac{r_n(\theta | \mathbf{y}) \left( \det \left| \nabla_\theta g (\nabla_\theta g)^T \right| \right)^{-1/2}  }{\mathbb{M}(r_n)} + \\
    &\hspace{0.8cm}  \zeta_n(0) \frac{r_n(\theta | \mathbf{y}) \left( \det \left| \nabla_\theta g (\nabla_\theta g)^T \right| \right)^{-1/2}  }{\mathbb{M}(r_n)}- \frac{ \phi_{\theta_0, nI(\theta_0)}(\theta) \left( \det \left| \nabla_\theta g (\nabla_\theta g)^T \right| \right)^{-1/2} }{ \mathbb{M}( \phi_{\theta_0, nI(\theta_0)})} \biggr| d \lambda(\theta)\\
   &\hspace{0cm}\leq \int_{\mathcal{M} } \frac{ r_n(\theta | \mathbf{y}) \left( \det \left| \nabla_\theta g (\nabla_\theta g)^T \right| \right)^{-1/2}  }{\mathbb{M}(r_n) } \left| 1 -  \zeta_n(0) \right| d\lambda(\theta) + I_n(0)\\
   &\hspace{0cm} \overset{P_{\theta_0}}{\to} 0.
\end{align*}

\subsection{Statement and Proof of Lemma \ref{lem:supportingLemma}}
\begin{lemma}\label{lem:supportingLemma}
Assume all assumptions from Lemma \ref{lem:BvMhwang}.  From the proof of Lemma \ref{lem:BvMhwang}, the term
\begin{equation}\label{eq:laplacesetuptubular}
 \mathcal{I}_n(\beta) = \frac{\int_{\mathcal{S}_{n}} \exp \left( \frac{-n||g(\theta_n(s) )||^2}{\beta} \right) \left| r_n^*(s | \mathbf{y}) - \phi^*_{0, I(\theta_0)}(s) \right| ds }{ \int_{\mathcal{S}_{n}}\exp\left( \frac{-n||g(\theta_n(s'))||^2}{\beta} \right) \phi^*_{0, I(\theta_0)}(s') ds' } ,
\end{equation}
goes to zero in the iterated limit as $\beta \to 0$ and $n \to \infty$.
\end{lemma}

\begin{proof}
Let $G(\theta) = \frac{1}{2} ||g(\theta)||^2$.  We begin by reviewing upper and lower bounds on this function that we will use in our Laplace approximation.  After performing the transformation to the tubular neighborhood, note that $G^\triangle(u,v)$ achieves its unique minimum (of zero) at $v=0$ for each $u$.  In the following, let $\dot{G}^\triangle$ and $\ddot{G}^\triangle$ be the gradient and Hessian, respectively, of $G^\triangle$ with respect to $v$.  Let $\epsilon >0$ be sufficiently small so that $\ddot{G}^\triangle(u,0)-\epsilon I$ is positive definite for all $u\in\mathcal{M}$.

Since $G^\triangle(u,0)=\dot{G}^\triangle(u,0) = 0$, the Taylor expansion around zero is
\[ G^\triangle(u,x) = \frac{1}{2} x' \ddot{G}^\triangle(u,0)x + R(u, x). \]
Let $\lambda_{\min}(u)$ be the smallest eigenvalue of $\ddot{G}^\triangle(u,0)$, which is greater than zero since this Hessian is positive-definite, and let $\lambda_{\min} = \min_{u \in \mathcal{M}} \lambda_{\min}(u)$.  Then, $G^\triangle(u,x) \geq \frac{1}{2} \lambda_{\min} ||x||_2 + R(u, \delta)$.  
The remainder term can be written in terms of the third partial derivatives of $G^\triangle$ \citep{hwang1980},
\[ R(u, x) =  \frac{1}{6} \sum_{i=1}^{d-t}\sum_{j=1}^{d-t}\sum_{k=1}^{d-t} \frac{\partial^3 G^\triangle (u,\bar{x}) }{\partial x_i \partial x_j \partial x_k} \cdot x_i \cdot x_j \cdot x_k, \]
where $\bar{x}$ is a line segment connecting $0$ and $x$.  Notice that this remainder term is bounded above by
\[ \frac{1}{6} \left(\max_{i,j,k} \max_{u \in \mathcal{M}, ||x||\leq \delta_1} \left|\frac{\partial^3 G^\triangle (u,\bar{x}) }{\partial x_i \partial x_j \partial x_k}\right| \right) \left| \sum_{i=1}^{d-t}\sum_{j=1}^{d-t}\sum_{k=1}^{d-t} x_i \cdot x_j \cdot x_k \right| =\frac{B}{6} ||x||_1^3 \leq \frac{B(d-t)^{3/2}}{6} ||x||_2^3, \]
where $B<\infty$ by the assumption that $g(\theta) \in C^3(\rr^d)$.  Thus, we can choose a $\delta_1$ sufficiently small such that if $||x||_2 \leq \delta_1$,
\begin{equation}\label{eq:Glambdabound}
    G^\triangle(u,x) \geq \frac{1}{2} \lambda_{\min} ||x||^2_2 > 0,
\end{equation}
and
\begin{equation}\label{eq:Gbothbound}
   \frac{1}{2} x' \ddot{G}^\triangle(u,0)x -  \epsilon ||x||^2_2 \leq G^\triangle(u,x) \leq \frac{1}{2} x' \ddot{G}^\triangle(u,0)x +  \epsilon ||x||^2_2.
\end{equation}
Note that $\delta_1$ depends on $u$; we suppress this dependence by taking a minimum over all possible $u$.

To prove this lemma, we will first prove that the denominator of \eqref{eq:laplacesetuptubular},
\begin{equation}\label{eq:tubulardenominator}
    n^{-t/2} \beta^{-t/2} \int_{\mathcal{S}_{n}}\exp\left( \frac{-nG(\theta_n(s'))}{\beta} \right) \phi^*_{0, I(\theta_0)}(s') ds',
\end{equation}
converges in iterated limit of $\beta \to 0$ and $n \to \infty$, where $\mathcal{S}_{n} = \{ s: n^{-1/2}s + \theta_0 \in T(\delta_2)\}$.  The result will then follow by an analogous argument on the numerator using the $\gamma_n$ function assumed to exist in the statement of Lemma \ref{lem:BvMhwang}.  Iterated limits are proven using the Moore-Osgood Theorem, which gives that the limits are interchangeable whenever one limit is pointwise and the other limit is uniform (see \cite{blake1946}).  

Transform \eqref{eq:tubulardenominator} from the local to the regular parameter, then transform from the regular parameter to the tubular neighborhood:
\begin{align*}
    \hspace{-0cm} n^{t/2} \beta^{-t/2} \int_{\psi(\mathcal{M})} \int_{||v|| \leq \delta_2 } \exp\left( \frac{-nG^\triangle(u,v) }{\beta} \right) \phi_{\theta_0, nI(\theta_0)}(\eta(u,v) ) J(u,v) dv ~ du,
\end{align*}
where $\eta$ is the transformation to the tubular neighborhood from the theorem statement and $\delta_2$ is sufficiently small so that $\eta$ is a diffeomorphism.  The properties of the change-of-variables Jacobian $J$ to the tubular neighborhood are reviewed in \cite{weyl1939}.  For our purposes, we will note that the Jacobian term $J$ is smooth in $v$ for each $u$, and that a sufficiently small $\delta_2$ will make it $\epsilon$-close to $1$. 

Perform an additional change of variables on \eqref{eq:tubulardenominator} using $s_1=n^{1/2}(\psi^{-1}(u)-\theta_0)$ and $s_2 = n^{1/2}v$:
\begin{align*}
    \hspace{-0cm} \beta^{-t/2} \int_{ S_n } \int_{||s_2|| \leq n^{1/2} \delta_2 }  \exp\left( \frac{-nG^\triangle( \psi(\theta_n(s_1)), n^{-1/2}s_2 ) }{\beta} \right) \phi^*_{0, I(\theta_0)}(s_1 + s_2 Q_\theta^\perp ) \hat{J}(s_1,s_2) ds_2 ~ ds_1,
\end{align*}
where $S_{n} = \{ s: n^{-1/2}s + \theta_0 \in \mathcal{M}\}$ and $\hat{J}(s_1,s_2) = J(\psi(\theta_n(s_1)), n^{-1/2}s_2)$.

Let $0<\delta_3<\delta_2$.  We split the integral in \eqref{eq:tubulardenominator} into the integral on a $T(\delta_3)$ tubular neighborhood, and the integral between this tube and the outer tube $T(\delta_2)$.  Then, we bound \eqref{eq:tubulardenominator} by $I_1 + I_2$, using the bounds in \eqref{eq:Gbothbound}, where
\begin{align*}
     &\hspace{-0.2cm}I_1= \beta^{-t/2} \int_{ S_n } \int_{||s_2|| \leq \delta_3 } \exp\left( \frac{ - n\left( s_2' \ddot{G}^\triangle(\psi(\theta_n(s_1)),0 )s_2 + \epsilon||s_2||_2^2 \right) }{n\beta} \right) \times \\
     &\hspace{7cm}\left(\phi^*_{0, I(\theta_0)}( s_1 + s_2Q_\theta^\perp ) \right) \hat{J}( s_1, s_2 ) ds_2 ~ ds_1 \\
     &\hspace{-0.2cm}I_2 = \beta^{-t/2} \int_{ S_n } \int_{ ||s_2||\in (\delta_3, n^{1/2} \delta_2] } \exp\left( \frac{- nG^\triangle( \psi(\theta_n(s_1)), n^{-1/2}s_2 ) }{\beta} \right)\times\\
     &\hspace{7cm}\left(\phi^*_{0, I(\theta_0)}( s_1 + s_2Q_\theta^\perp ) \right) \hat{J}( s_1, s_2 ) ds_2 ~ ds_1.
\end{align*}
For sufficiently small $\delta_2$, the integral $I_1$ is further bounded above by
\begin{align*}
    &\hspace{-0cm} I_1 \leq \beta^{-t/2} \int_{ S_n } \sup_{||s_2^*|| \leq \delta_3 } \left(\phi^*_{0, I(\theta_0)}( s_1 + s_2^* Q_\theta^\perp ) \right) (1 + \epsilon)  \times \\
     &\hspace{2cm}  
     \int_{||s_2|| \leq \delta_3 } \exp\left( \frac{ - n\left( s_2' \ddot{G}^\triangle(\psi(\theta_n(s_1)),0 )s_2 + \epsilon||s_2||_2^2 \right) }{n\beta} \right)ds_2 ~ ds_1. 
\end{align*} 

Consider the change of variables\footnote{The matrix square root $B = A^{1/2}$ is such that $BB = A$.} $(s_1, w) = \left(s_1, \left( \ddot{G}^\triangle(\psi(\theta_n(s_1)),0 )+\epsilon^2I)/(2\beta)\right)^{1/2} s_2\right)$, and let 
\[W^+_{n,s_1,\beta} = \left\{ w\in \rr^{t}: ||\left( \ddot{G}^\triangle(\psi(\theta_n(s_1)),0 ) + \epsilon^2 I\right)^{-1/2}w|| \leq \frac{\delta_3}{\sqrt{2\beta}}\right\}.\]
When written in terms of $w$, the inner integral is a normal CDF symmetric about zero. So, for $\mathbb{W} \sim \mathcal{N}_{t}(0,I_t)$, we rewrite the current upper bound on $I_1$ as,
\begin{equation} \label{eq:ionebound}
\hspace{-0.3cm} \int_{S_n} \sup_{||s_2^*|| \leq \delta_3 }\phi^*_{0, I(\theta_0)}( s_1 + s_2^* Q_\theta^\perp )  (1 + \epsilon) \left|  \frac{ \left( \ddot{G}^\triangle(\psi(\theta_n(s_1)),0 ) + \epsilon^2 I\right) }{ 2\pi } \right|^{-1/2} P\left(\mathbb{W} \in W^+_{n,s_1,\beta}\right) ~ ds_1.
\end{equation}

Let $C_{\theta_0, \epsilon}$ be a compact subset of $S_n$ such that $0 \in C_{\theta_0, \epsilon}$.  We select a sufficiently large $C_{\theta_0, \epsilon}$ so that all values $s_1 \in S_n\backslash C_{\theta_0,\epsilon}$ are significantly pushed into tails of the normal density $\phi^*_{0, I(\theta_0)}$ to get the further bound on $I_1,$
\begin{align*}
    &\hspace{-0.2cm} \epsilon +  \int_{S_n \cap C_{\theta_0, \epsilon} } \sup_{||s_2^*|| \leq \delta_3 }\phi^*_{-s_2^* Q_\theta^\perp, I(\theta_0)}( s_1 )  (1 + \epsilon) \left|  \frac{ \left( \ddot{G}^\triangle(\psi(\theta_n(s_1)),0 ) + \epsilon^2 I\right) }{ 2\pi } \right|^{-1/2} \times \\
    &\hspace{7cm}P\left(\mathbb{W} \in W^+_{n,s_1,\beta}\right) ~ ds_1.
\end{align*} 
On $S_n \cap C_{\theta_0, \epsilon}$, the quadratic form in the normal density is uniformly continuous.  Thus, for sufficiently small $\delta_3$,
\vspace{-0.3cm}
\begin{equation}\label{eq:bound_on_denominator_approx}
    \hspace{-0cm}I_1 \leq \epsilon + \int_{S_n} \phi^*_{0, I(\theta_0)}( s_1 ) \exp\left(\epsilon \right) (1 + \epsilon) \left|  \frac{ \left( \ddot{G}^\triangle(\psi(\theta_n(s_1)),0 ) + \epsilon^2 I\right) }{ 2\pi } \right|^{-1/2} ds_1.
\end{equation}

For the $I_2$, use the lower bound in \eqref{eq:Glambdabound} to note that
\[-G^\triangle (\psi(\theta_n(s_1)), n^{-1/2}s_2) \geq -n^{-1}\lambda_{\min}\delta_3>0. \] 
Thus, $I_2$ is bounded above by
\[ \hspace{-0cm} \beta^{-t/2} \exp\left( \frac{- \lambda_{\min}\delta_3 }{\beta} \right) \int_{ S_n }   \int_{ \delta <||s_2|| \leq n^{1/2} \epsilon }  \left(\phi^*_{0, I(\theta_0)}( s_1 + s_2Q_\theta^\perp ) \right) \hat{J}( s_1, s_2 ) ds_2 ~ ds_1.  \]
Since the integral term is just a change of variables on a normal density, this value is further bounded above by $\beta^{-t/2} \exp\left(-\lambda_{\min}\delta_3/\beta \right),$ which is without a dependence on $n$.  Thus, $I_2$ goes to zero as $\beta \to 0$, pointwise in $n$ and uniformly over $\beta \in (0,1)$. 

We have now proven that \eqref{eq:tubulardenominator} is bounded above by \eqref{eq:bound_on_denominator_approx} in the iterated limit.  Assume that $\delta_3$ and $\delta_1$ are each taken to be $\min\{\delta_3, \delta_1\}$.  Since $\epsilon$ was chosen arbitrarily, we tighten this upper bound to,
\begin{equation}\label{eq:upper_bound_on_denominator}
    \int_{S_n} \left(\phi^*_{0, I(\theta_0)}( s_1 ) \right) \left|  \frac{ \left( \ddot{G}^\triangle(\psi(\theta_n(s_1)),0 ) I\right) }{ 2\pi } \right|^{-1/2} ds_1.
\end{equation}
We will now prove an identical lower bound on \eqref{eq:tubulardenominator}.

Following the same process used to get the upper bound, we derive the following lower bound on \eqref{eq:tubulardenominator}:
\begin{align}\label{eq:lower_bound_1}
& \int_{S_n} \inf_{||s_2^*||_2 \leq \delta_3 } \left(\phi^*_{-s_2^* Q_\theta^\perp, I(\theta_0)}( s_1 ) \right)  (1 - \epsilon) \left|  \frac{ \left( \ddot{G}^\triangle(\psi(\theta_n(s_1)),0 ) - \epsilon^2 I\right) }{ 2\pi } \right|^{-1/2} \times \\
&\hspace{7cm} P\left(\mathbb{W} \in W^-_{n,s_1,\beta}\right) ~ ds_1, \nonumber
\end{align}
where
\[W^-_{n,s_1,\beta} = \left\{ w\in \rr^{t}: ||\left( \ddot{G}^\triangle(\psi(\theta_n(s_1)),0 ) - \epsilon^2 I\right)^{-1/2}w||_2 \leq \frac{\delta_3}{\sqrt{2\beta}}\right\}. \]
Notice that $W^-_\beta \subseteq W^-_{n,s_1,\beta}$ where
\[W^-_\beta = \left\{ w\in \rr^{t}: ||w||_2 \leq \frac{(\lambda_{\min}-\epsilon)\delta_3}{\sqrt{2\beta}}\right\}.\]
We then bound \eqref{eq:lower_bound_1} below by 
\begin{align*}
& P\left(\mathbb{W} \in W^-_\beta\right) \int_{S_n } \inf_{||s_2^*||_2 \leq \delta_3 } \phi^*_{-s_2^* Q_\theta^\perp, I(\theta_0)}( s_1 ) (1 - \epsilon) \left|  \frac{ \left( \ddot{G}^\triangle(\psi(\theta_n(s_1)),0 ) - \epsilon^2 I\right) }{ 2\pi } \right|^{-1/2} ~ ds_1.
\end{align*}
Note that $P\left(\mathbb{W} \in W^-_\beta\right)$ is the probability mass of a non-empty set around the mean of a multivariate normal distribution, and is therefore bounded away from zero.  For each $n$, this limit goes to 
\[ \hspace{-2cm} \int_{S_n } \inf_{||s_2^*||_2 \leq \delta_3 } \phi^*_{-s_2^* Q_\theta^\perp, I(\theta_0)}( s_1 )(1 - \epsilon) \left|  \frac{ \left( \ddot{G}^\triangle(\psi(\theta_n(s_1)),0 ) - \epsilon^2 I\right) }{ 2\pi } \right|^{-1/2} ~ ds_1. \]
\[ \geq -\epsilon + \int_{S_n\cap C_{\theta_0, \epsilon}} \left( \phi^*_{0, I(\theta_0)}( s_1 )(1 - \epsilon) \exp\left(-\epsilon \right) \right) \left|  \frac{ \left( \ddot{G}^\triangle(\psi(\theta_n(s_1)),0 ) - \epsilon^2 I \right) }{ 2\pi } \right|^{-1/2} ds_1 \]
as $\beta\to 0$.  Since the dependence on $n$ and $\beta$ are separated, the limit in $n$ is uniform in $\beta$.  Thus, the conditions of the Moore-Osgood Theorem are satisfied.  

Assume that $\delta_3$ and $\delta_1$ are each taken to be $\min\{\delta_3, \delta_1\}$.  Since $\epsilon$ was chosen arbitrarily, we tighten this bound to
\[ \int_{S_n} \left( \phi^*_{0, I(\theta_0)}( s_1 )  \right) \left|  \frac{ \left( \ddot{G}^\triangle(\psi(\theta_n(s_1)),0 ) \right) }{ 2\pi } \right|^{-1/2} ds_1, \]
which is the same as the upper bound.

We will now establish an upper bound for the numerator of \eqref{eq:laplacesetuptubular} and show that it goes to zero in the iterated limit.  For sufficiently large $n$, the numerator is bounded above by
\begin{equation*}
    n^{-t/2} \beta^{-t/2} \int_{\mathcal{S}_{n}}\exp\left( \frac{-nG(\theta_n(s'))}{\beta} \right) \exp\left(-\gamma_{n} \right) ds',
\end{equation*}
and the steps used to bound the denominator of \eqref{eq:laplacesetuptubular} follow when $\exp(-\gamma_n)$ replaces $\phi_{0, I(\theta_0)}$.  Using the uniform equicontinuity and integrability assumptions on $\gamma_n$, this integral is bounded above by
\[ \int_{S_n} \exp\left( -\gamma_n( s_1 ) \right) (1 + \epsilon) \exp\left( \epsilon \right) \left|  \frac{ \left( \ddot{G}^\triangle(\psi(\theta_n(s_1)),0 ) + \epsilon^2 I \right) }{ 2\pi } \right|^{-1/2} ds_1. \]
Using the fact that $\ddot{G}^\triangle(u,0)$ is continuous, we conclude that the numerator of \eqref{eq:laplacesetuptubular} is further bounded above by
\[ \left(\max_{u \in \mathcal{M}} \left|  \frac{ \left( \ddot{G}^\triangle(u, 0) + \epsilon^2 I \right) }{ 2\pi } \right|^{-1/2} \right) (1 + \epsilon) \exp\left( \epsilon \right) \int_{S_n} \exp \left( -\gamma_n( s_1 )  \right) d s_1, \]
which goes to zero in the iterated limit since it is without a dependence on $\beta$.  Thus,
\begin{equation}
 \lim_{\beta \to 0} \lim_{n \to \infty} \frac{n^{-t/2}\beta^{-t/2} }{ n^{-t/2}\beta^{-t/2} } \mathcal{I}_n(\beta) =  \lim_{n\to\infty} \lim_{\beta \to 0} \frac{n^{-t/2}\beta^{-t/2}}{ n^{-t/2}\beta^{-t/2} } \mathcal{I}_n(\beta) = 0
\end{equation}
\end{proof}

\subsection{Proof of Theorem \ref{thm:ConstrainedVonMises}}
For compact $\mathcal{M}$, this result is an immediate application of the triangle inequality. Assume non-compact $\mathcal{M}$.  On the sets $A_1$ and $A_2$ from Lemma \ref{lem:HwangVCGFD}, the proof reduces to the compact case.  On $A_3$, we have shown in Appendix \ref{a:Lem1Proof} that the Constrained GFD goes to zero.  In the case of the normal density,
\begin{align*}
  \frac{\int_{A_3} \frac{\sqrt{\det | I(\theta_0)| }}{\sqrt{2\pi}} e^{- s^T I(\theta_0) s / 2}  \left( \det \left| \nabla_\theta g (\nabla_\theta g)^T \right| \right)^{-1/2} d\lambda(\theta_n(s))}{\int_{S_n }  \phi^*_{0, I(\theta_0)}(s)\left( \det \left| \nabla_\theta g (\nabla_\theta g)^T \right|\right)^{-1/2} d s}
\end{align*}
for $S_n = \{ s : \theta_0 + n^{-1/2}s \in \mathcal{M}\}$.  The numerator goes to zero since $\min_{A_3}||s||\to \infty$ for $\mathcal{M}$ non-compact and $\nabla_\theta g$ everywhere bounded by assumption.  Meanwhile, the denominator must be bounded above zero since $0 \in S_n$ always, and $\left( \det \left| \nabla_\theta g (\nabla_\theta g)^T \right| \right)^{-1/2}$ is non-zero on the manifold.

\subsection{Proof of Corollary \ref{cor:forIID}}
Let $l(M) = M'M$ and $A_i$ be the DGA for the $i^{th}$ observation $y_i$.  Note that 
\[ [l( \nabla_\theta A_i(w,\theta))]_{j,k} \leq \sup_{\theta \in \mathcal{M}} \max_{j'\in[d], k' \in [d]}  |l( \nabla_\theta A_i(w,\theta))|_{j',k'} \leq  C_1 \zeta(y_i), \]
\[ [l( \nabla_\theta A_i(w,\theta) P_\theta )]_{j,k} \leq \sup_{\theta \in \mathcal{M}} \max_{j'\in[d], k' \in [d]} |l( \nabla_\theta A_i(w,\theta) P_\theta)|_{j',k'} \leq C_2 \zeta(y_i), \]
by the assumption on the Spectral Norm and for some $C_1,C_2 \in \rr^+$, where $[M]_{j,k}$ is the $(j,k)^{\text{th}}$ element of a matrix $M$.
Thus, 
\[ \frac{1}{n^d}D\left(\nabla_\theta A(w,\theta)\right)^2 = \det\left(\frac{1}{n} \sum_{i=1}^n l(\nabla_\theta A_i(w, \theta)) \right)  \to \pi^2_1(\theta)\]
\[ \frac{1}{n^{d-t}}D^\star\left(\nabla_\theta A(w,\theta)P_\theta \right)^2 = \pdet\left(\frac{1}{n} \sum_{i=1}^n  l(\nabla_\theta A_i(w, \theta)P_\theta) \right)  \to \pi^2_2(\theta)\]
as $n \to \infty$ by the Uniform Law of Large Numbers \citep{newey1994} and the continuity of $P_\theta$.

\end{appendix}

\bibliographystyle{rss}
\bibliography{examples}    
\end{document}